\documentclass[10pt, twocolumn]{IEEEtran}

\usepackage{times,amsmath}
\usepackage{amsfonts}
\usepackage{amssymb,bm}
\usepackage{bbm}
\usepackage{algorithm}
\usepackage{algpseudocode}
\usepackage{cite}
\usepackage{graphicx}
\usepackage{epstopdf}
\usepackage[caption=false,font=footnotesize]{subfig}
\usepackage{fixltx2e}
\usepackage{balance}
\usepackage{cases}
\usepackage{amsthm}

\newtheorem{lemma}{Lemma}

\newtheorem{assumption}{Assumption}
\newtheorem{proposition}{Proposition}
\usepackage{chngcntr}
\counterwithout{theorem}{section}
\counterwithout{definition}{section}
\counterwithout{lemma}{section}
\counterwithout{remark}{section}
\counterwithout{assumption}{section}
\counterwithout{proposition}{section}

\DeclareMathOperator{\E}{\mathbbmss{E}}

\DeclareMathOperator{\Trn}{\mathsf{T}}
\DeclareMathOperator{\Hrm}{\mathsf{H}}

\allowdisplaybreaks

\begin{document}
%%% Consecutive citation with same author names
\bstctlcite{IEEEexample:BSTcontrol}
%\title{Max-Min Fair Multigroup Multicast Beamforming Strategies for Overloaded Transmission}
\title{Rate-Splitting for Max-Min Fair Multigroup Multicast Beamforming in Overloaded Systems}
\author{Hamdi~Joudeh and Bruno~Clerckx
\thanks{This work is partially supported by the UK EPSRC under grant EP/N015312/1.
A preliminary version of this paper was presented at the 17th  IEEE International workshop on Signal Processing  advances in Wireless Communications
(SPAWC), Edinburgh, UK, July 2016.}
\thanks{The authors are with the Communications and Signal Processing Group,
Department of Electrical and Electronic Engineering, Imperial College, London SW7 2AZ, U.K. (email: hamdi.joudeh10@imperial.ac.uk; b.clerckx@imperial.ac.uk).}
}
\maketitle

\begin{abstract}
In this paper, we consider the problem of achieving max-min fairness amongst multiple co-channel multicast groups through transmit beamforming.
We explicitly focus on overloaded scenarios in which the number of transmitting antennas is insufficient to neutralize all inter-group interference.
Such scenarios are becoming increasingly relevant in the light of growing low-latency content delivery demands, and also commonly appear in multibeam satellite systems.
We derive performance limits of classical beamforming strategies using DoF analysis unveiling their limitations; for example, rates saturate in overloaded scenarios due to inter-group interference.
To tackle interference, we propose a strategy based on degraded beamforming and successive interference cancellation.
While the degraded strategy resolves the rate-saturation issue, this comes at a price of sacrificing all spatial multiplexing gains.
This motivates the development of a unifying strategy that combines the benefits of the two previous strategies.
We propose a beamforming strategy based on rate-splitting (RS) which divides the messages intended to each group into a degraded part and a designated part, and transmits a superposition of both degraded and designated beamformed streams.
The superiority of the proposed strategy is demonstrated through DoF analysis.
Finally, we solve the RS beamforming design problem and demonstrate significant performance gains through simulations.
\end{abstract}
%\begin{IEEEkeywords}
%.
%\end{IEEEkeywords}
\IEEEpeerreviewmaketitle
%%%
\section{Introduction}
\label{Section_Introduction}
Physical layer multicasting in wireless networks has received considerable research attention in recent years.
In the simplest multicasting scenario, a transmitter communicates a common message to a group of receivers \cite{Sidiropoulos2006}.
More complex scenarios involve the simultaneous transmission of distinct messages to multiple multicast groups, known as multigroup multicasting \cite{Karipidis2008}.
Such scenarios are likely to occur in future wireless networks due to the emergence of content-oriented services and wireless caching.
Indeed, multicast groups are naturally (or artificially!) formed by users requesting similar content (or coded content), opening the door for an increased role of physical layer multicasting solutions \cite{Tao2016,Maddah-Ali2016}.
Multicasting scenarios also appear in multibeam satellite systems due to standardized framing structures, where each data stream 
accommodates the requests of multiple users  \cite{Christopoulos2013,Christopoulos2015a,Vazquez2016,Joroughi2017}.

Despite the fact that simple isotropic transmission is sufficient in the single-group multicasting case, the employment of multiantenna signal processing techniques was shown to achieve nontrivial performance gains \cite{Sidiropoulos2006,Jindal2006a}.
Amongst the various multiantenna signalling solutions, the most popular ones are those with unity-rank input covariance matrices (beamforming).
Such solutions allow the use of the well-established scalar codes, originally designed for single-antenna systems, hence greatly simplifying the channel coding problem \cite{Jafar2004,Tse2005}.
Although the single-group multicasting setup may be considered relatively simple, the problem of finding the optimum beamforming direction was in fact shown to be NP-hard \cite{Sidiropoulos2006}.
The authors of \cite{Sidiropoulos2006} proposed an approximate method using Semidefinite Relaxation (SDR) and  randomization.
Moving towards multigroup multicasting, beamforming becomes more crucial as such scenarios are fundamentally limited by inter-group interference.
In general, multigroup multicast beamforming design problems inherit the difficulty of single-group problems, while posing the additional challenge of managing inter-group interference.
Different beamforming designs yield different tradeoffs between the rates that can be simultaneously supported for different multicast groups.
In \cite{Karipidis2008}, the SDR-randomization approach of  \cite{Sidiropoulos2006} is extended to solve two multigroup multicast beamforming problems: the Quality of Service (QoS) constrained power minimization problem and the power constrained Max-Min Fair (MMF) problem.
Alternative solutions based on convex approximation methods were later proposed, exhibiting marginally improved performances in certain scenarios, yet lower complexities \cite{Bornhorst2011,Schad2012}.
Moreover, multigroup multicast beamforming problems were also extended to many other scenarios including
per-antenna power constrained transmission \cite{Christopoulos2014}, large-scale antenna arrays \cite{Christopoulos2015}, multi-cell coordination \cite{Hong2015}, relay networks \cite{Bornhorst2012} and cache aided networks \cite{Tao2016}, just to name a few.
\emph{Motivation}: All aforementioned works adopt a classical beamforming framework in which one designated data stream is beamed to each group,
and receivers decode their corresponding messages while treating all other data streams as noise.
While this strategy can neutralize inter-group interference under a sufficient number of transmitting antennas,
it fails to do so in overloaded scenarios with a relatively high number of co-scheduled groups/users \cite{Joudeh2016c}.
This can be roughly attributed to the number of data streams exceeding the number of spatial Degrees of Freedom (DoF) created through classical beamforming.
Overloaded scenarios are becoming more and more relevant in the light of the growing demands for ultra-low latency and ultra-high connectivity \cite{Dai2015a,Piovano2016}.
Moreover, such scenarios commonly arise in multibeam satellite systems where messages intended to different users are embedded in each data stream.
Users decoding the same stream (or frame) form a multicast group, and multiple of such groups are usually co-scheduled in an overloaded manner \cite{Christopoulos2013,Christopoulos2015a,Vazquez2016,Joroughi2017}.
In the existing literature on multigroup multicast beamforming,
overloaded transmissions have been implicitly considered through simulations with system parameters that correspond to such scenarios \cite{Karipidis2008,Bornhorst2011,Schad2012,Christopoulos2015a}.
However, a comprehensive analysis and explicit treatment of interference in overloaded scenarios is absent.
In this work, we focus on the MMF problem in overloaded multigroup multicasting scenarios and explore the potentials of applying
alternative beamforming strategies.
%%%
Next, we give an overview of the main contributions of this paper.
\emph{Contributions}:
Considerable attention has been given to solving various multigroup multicast beamforming problems; however,
little has been done on the derivation of performance limits which have been mainly analysed through extensive simulations.
First, we make progress towards understanding the performance limits of the classical beamforming framework by
characterizing the MMF-DoF performance.
The MMF-DoF is a first-order approximation of the MMF-rate in the high SNR regime, which is roughly interpreted as
the maximum fraction of an interference-free  message that can be simultaneously communicated to all multicast groups.
The tractability of the DoF metric stems from its independence of varying system parameters such as the transmitting power and
channel realization.
Alternatively, it captures the interference-management capabilities as a function
of  the fixed system parameters, i.e.
the number of transmitting antennas, multicast groups and users in each group.
From the MMF-DoF, we identify the conditions under which the system is deemed overloaded,
and gain insights into the MMF-rate performance. For example, it reflects the
performance saturation at high SNRs in fully-overloaded scenarios.
Second, we depart from the classical beamforming paradigm of treating inter-group interference as noise,
and propose a beamforming strategy that incorporates Successive Interference Cancellation (SIC).
Multicast groups are ordered such that receivers decode messages (and cancel interference) in a successive manner;
a reminiscence of the Non-Orthogonal Multiple Access (NOMA) scheme proposed for multi-user beamforming in \cite{Hanif2016}.
The relevance of non-orthogonal schemes in overloaded scenarios follows from the fact that the number of messages communicated in the same resource block (time/frequency) exceeds the number of spatial DoF.
This non-orthogonal strategy degrades the channel, as at least one receiver ends up
decoding all messages, limiting the sum-DoF to unity.
An important implication is that while saturating MMF-rate performances are avoided,
all spatial multiplexing gains are in fact annihilated.
Hence, this strategy can be approximated (in a DoF sense) by a degraded single-beam strategy.

Third, we propose a generalized strategy that bridges the gap between the designated (classical) and degraded beamforming
strategies.
The proposed strategy is formulated in terms of Rate-Splitting (RS), where each message is divided into a degraded part and a designated part.
The transmitted signal is therefore a superposition of degraded and designated beamformed data streams.
RS has been recently applied in a variety of multiuser beamforming scenarios with uncertain channel state information at the transmitter \cite{Clerckx2016}\footnote{Contrary to the multigroup multicasting scenario considered in this paper, \cite{Clerckx2016} (and references therein) is mainly focused on conventional multiuser scenarios with unicast transmissions.}.
%%%
We show that the RS beamforming strategy brings significant gains to overloaded multigroup multicasting scenarios by deriving its MMF-DoF performance, and unveiling its strict superiority to the two preceding strategies.
RS exploits partial gains achieved through spatial multiplexing while maintaining a non-saturating performance through the degraded part, and goes beyond simply switching between both.
Fourth, we solve the RS MMF beamforming design problem by invoking the Weighted Minimum-Mean Square Error (WMMSE) approach \cite{Christensen2008,Shi2011}, which is particularly suitable for the problem due to the sum-rate expressions arising from rate-splitting.
Moreover, the performance gains achieved by employing the proposed beamforming strategy in overloaded scenarios are illustrated through simulation result.

The employment of RS in multigroup multicast beamforming was first proposed in a preliminary version of this paper, which can be found in \cite{Joudeh2016c}.
To simplify the DoF analysis, equally sized multicasting groups are assumed in \cite{Joudeh2016c}.
Moreover, \cite{Joudeh2016c} only provides a trivial lowerbound for the MMF-DoF of the RS strategy, which is shown to be loose through simulations.
In this paper, we pose no restrictions on the sizes of multicasting groups
and derive an exact characterization of the MMF-DoF performance achieved through the RS strategy.
\emph{Organization}: Section \ref{Section_System_Model} describes the general multigroup multicasting system model.
Classical beamforming is discussed in Section \ref{Section_NoRS_Problem}, in which we also define the DoF metric and derive the MMF-DoF performance of the classical strategy.
In Section \ref{Section_Degraded_Beamforming}, the NOMA inspired degraded beamforming strategy is proposed and its  MMF-DoF is characterized.
The RS strategy is proposed in Section \ref{Section_RS}, where its  MMF-DoF performance is also derived alongside some insights into how the MMF-DoFs of different strategies compare.
In Section \ref{Section_Optimization}, a WMMSE algorithm is developed to optimize the RS MMF beamforming strategy.
Simulation results are presented in Section \ref{Section_Numerical_Results} and Section \ref{Section_Conclusion} concludes the paper.

\emph{Notation}:
The following notations are used in the paper.
$a,A$ are scalars, $\mathbf{a}$ is a column vector, $\mathbf{A}$ is a matrix and $\mathcal{A}$ is a set.
$(a_{1},\ldots,a_{K})$ denotes a $K$-tuple of scalars, which is also represented by a column vector $\mathbf{a}$.
The superscrips $(\cdot)^{\Trn}$ and $(\cdot)^{\Hrm}$ denote the transpose and conjugate-transpose respectively.
$\mathrm{tr}(\cdot)$, $\|\cdot\|$ and $\E\{\cdot\}$ are the trace, Euclidian norm and expectation operators respectively.
%%%
$\mathrm{dim}(\cdot)$  and $\mathrm{null}(\cdot)$ refer to the dimension and the null space respectively.
%%
%%%
\section{System Model}
\label{Section_System_Model}
We consider a wireless system comprising a single transmitter equipped with $N$ antennas and $K$ single-antenna receivers indexed by the set $\mathcal{K} \triangleq \{ 1, \ldots, K \}$.
Receivers are grouped into the $M$ $(1 \leq M \leq K)$ multicast groups $\mathcal{G}_{1},\ldots,\mathcal{G}_{M}$, where $\mathcal{G}_{m}$ is the set of receivers belonging to the $m$th group, $m \in \mathcal {M}$, and $\mathcal{M} \triangleq \{ 1, \ldots, M\}$ is the index set of all groups.
Such grouping is carried out based on content, i.e. receivers belonging to the same group are interested in the same message.
It is assumed that each receiver belongs to exactly one group. Thus $\bigcup_{m\in\mathcal{M}} \mathcal{G}_{m} = \mathcal{K}$
and $\mathcal{G}_{m} \cap \mathcal{G}_{j} =\emptyset, \forall m,j \in \mathcal{M}$ and
$m\neq j$.
Denoting the size of the $m$th group by $G_{m} = |\mathcal{G}_{m}| $, it is assumed without loss of generality that group sizes are in an ascending order, i.e.
\begin{equation}
\label{Eq_Group_order}
G_{1} \leq G_{2} \leq \ldots \leq G_{M}.
\end{equation}
To map users to their respective groups, we define $\mu: \mathcal{K} \rightarrow \mathcal{M}$ such that
$\mu(k) = m$ for all $k \in \mathcal{G}_{m}$.

Let $\mathbf{x} \in \mathbb{C}^{N}$ denote the signal vector transmitted in a given channel use.
This transmitted signal is subject to an average power constraint $\E \left\{ \mathbf{x}^{\Hrm}\mathbf{x} \right\} \leq P$, where $P  > 0$.
Denoting the corresponding signal received by the $k$th user as $y_{k}$, the input-output relationship is given as
\begin{equation}
y_{k} = \mathbf{h}_{k}^{\Hrm} \mathbf{x} + n_{k}
\end{equation}
where $\mathbf{h}_{k}\in \mathbb{C}^{N}$ is the channel vector between the transmitter and the $k$th user, and $n_{k} \in \mathcal{CN}(0,\sigma_{\mathrm{n},k}^{2})$ is the receiver Additive White Gaussian Noise (AWGN).
The channel matrix composed of $K$ channel vectors is given by
$\mathbf{H} \triangleq \left[\mathbf{h}_{1} \; \cdots \; \mathbf{h}_{K}\right]$.
It is assumed without loss of generality that $\sigma_{\mathrm{n},1}^{2},\ldots,\sigma_{\mathrm{n},K}^{2} = \sigma_{\mathrm{n}}^{2}$, from which the transmit SNR is given by $P/\sigma_{\mathrm{n}}^{2}$.
Moreover, it is further assumed that the transmitter perfectly knows $\mathbf{H}$ and each receiver knows its own channel.
In multigroup multicasting, the transmitter wishes to communicate the distinct messages
$W_{1},\ldots,W_{M}$ to users in $\mathcal{G}_{1},\ldots,\mathcal{G}_{M}$ respectively.
Messages are mapped to the transmitted signal as $W_{1},\ldots,W_{M} \mapsto \mathbf{x}$ using some encoding function.
On the other end of the channel, messages are decoded from the received signals as
$y_{k} \mapsto \hat{W}_{\mu(k)}^{k}$, where $\hat{W}_{\mu(k)}^{k}$ is the $k$th user's estimate of $W_{\mu(k)}$.
For a given strategy, the group-rates denoted by $r_{1},\ldots,r_{M}$ correspond to the maximum rates of communicating $W_{1},\ldots,W_{M}$ respectively, while guaranteeing successfully decoding (with high probability) by all corresponding users.
Transmission strategies can be designed to optimize various objective functions subject to different constraints.
Here we are interested in MMF designs which aim to maximize the symmetric rate that can be simultaneously achieved by all groups subject to a power constraint.
Moreover, channel coding is abstracted by assuming Gaussian codes and the focus is on designing and analysing the beamforming schemes.
In this context, it is worth mentioning that from an information theoretic point of view, the setup at hand
is modeled by a compound MISO broadcast channel \cite{Maddah-Ali2010,Gou2011}.
For a class of overloaded scenarios, it was shown that the optimum sum-DoF is achieved through interference alignment over rational dimensions by exploiting the signal level (or power) domain.
While the employment of the power domain features in this work, we focus on a class of strategies based on beamforming.
This is more inline with current deployments of multiuser MIMO systems \cite{Castaneda2017}.
Moreover, we consider one-shot transmission schemes with no time-sharing between strategies.
This is suitable for systems with rigid scheduling and/or tight latency constraints, for example as in multibeam satellite systems \cite{Christopoulos2013,Christopoulos2015a,Vazquez2016,Joroughi2017}, and also allows for simpler designs.
%%%
%%%
\section{Designated (Classical) Beamforming}
\label{Section_NoRS_Problem}
In classical beamforming, the $M$ messages are first mapped into independent designated symbol streams as
$W_{1},\ldots,W_{M} \mapsto  s_{1},\ldots,s_{M}$, which are then beamformed as
\begin{equation}
\label{Eq_x_NoRS}
\mathbf{x} = \sum_{m=1}^{M} \mathbf{p}_{m}s_{m}
\end{equation}
where $\mathbf{p}_{m} \in \mathbb{C}^{N}$ denotes the $m$th group's designated beamforming vector.
Defining $\mathbf{s} \triangleq [s_{1} \; \cdots \; s_{M}]^{\Trn}$ and assuming that $\E\{\mathbf{s} \mathbf{s}^{\Hrm} \} = \mathbf{I}$,
the transmit power constraint under beamforming reduces to $\sum_{m=1}^{M} \|\mathbf{p}_{m}\|^{2}  \leq P$.
In some analysis, it helps to emphasize the structure of the beamformers. Hence, we write
$\mathbf{p}_{m} = \sqrt{q_{m}} \mathbf{w}_{m}$, where $q_{m} = \|\mathbf{p}_{m}\|^{2}$ is the power allocated to the $m$th beam
and $\mathbf{w}_{m}$ is a unit vector that denotes the beamforming direction.
%%
%It is evident that multigroup multicast beamforming reduces to single-group multicast beamforming when $M=1$ and multiuser beamforming when $M=K$.
%%

%%
The Signal to Interference plus Noise Ratio (SINR) experienced by the $k$th user is given by
\begin{equation}
\gamma_{k}  = \frac{|\mathbf{h}_{k}^{\Hrm}\mathbf{p}_{\mu(k)}|^{2}}{\sum_{m\neq \mu(k)} |\mathbf{h}_{k}^{\Hrm}\mathbf{p}_{m}|^{2} + \sigma_{\mathrm{n}}^{2}}.
\end{equation}
The achievable rate from the $k$th user's point of view under Gaussian signalling is given by $R_{k} = \log_{2}(1 + \gamma_{k} )$.
Since the $m$th data stream carries a message that should be decoded by all users in $\mathcal{G}_{m}$, the corresponding code-rate should not exceed the rate achievable by the \emph{weakest} receiver in the group.
Hence, the group-rate is given by
\begin{equation}
\label{Eq_group_rate}
r_{m} = \min_{i \in \mathcal{G}_{m}} R_{i}.
\end{equation}
%%%
\subsection{Achieving Max-Min Fairness}
The objective of the MMF design is to achieve fairness among groups subject to a transmit power constraint.
The classical beamforming MMF problem is formulated as
%%
%%%%
\begin{equation}
\label{Eq_Problem_R_NoRS}
\mathcal{R}(P):
\begin{cases}
       \underset{\mathbf{P}}{\max} & \underset{m \in \mathcal{M}}{\min} \
        \underset{i \in \mathcal{G}_{m}}{\min} \ R_{i}  \\
       \text{s.t.}  & \displaystyle{\sum_{m=1}^{M} \|\mathbf{p}_{m}\|^{2} \leq P}
\end{cases}
\end{equation}
%%%%
where $\mathbf{P} \triangleq [\mathbf{p}_{1} \; \cdots \; \mathbf{p}_{M}]$ is the beamforming (or precoding) matrix. The inner minimization in \eqref{Eq_Problem_R_NoRS} accounts for the multicast nature within each group as shown in \eqref{Eq_group_rate}. On the other hand, the outer minimization accounts for the fairness across groups.
It is common practice to formulate the beamforming problem in terms of the SINRs which is equivalent to \eqref{Eq_Problem_R_NoRS} due to the  one-to-one monotonic relationship between $R_{i}$ and $\gamma_{i}$.
Problems are formulated in terms of the achievable rates in this work to facilitate the DoF performance analysis.
In this section we characterize the performance of classical beamforming through DoF analysis.
The relevance of such analysis follows from the fact that achieving max-min fairness requires a simultaneous increase in powers allocated to all streams as $P$ increases.
%%%
In scenarios where $N$ is sufficient to place each beam in the null space of all its unintended groups, each multicast group receives an interference free stream.
However, if such condition is violated, the transmission becomes interference limited and DoF analysis can help us gain insight into the performance.
%%
%%%%
\subsection{Degrees of Freedom}
%%%%
To facilitate the definition of the DoF metric, we start by defining a beamforming scheme as a family of feasible beamforming matrices, with one for each SNR (or power) level.
%%%
This is denoted by $\left\{ \mathbf{P} (P) \right\}_{P}$, where $\mathbf{P}(P)$ is associated with the $P$th level and adheres to the power constraint.
%%%
A beamforming scheme is associated with a set of achievable user-rate tuples given by
$\left\{ \left( R_{1}(P),\ldots,R_{K}(P) \right) \right\}_{P}$, where
$\left( R_{1}(P),\ldots,R_{K}(P) \right)$ is the user-rate tuple achieved by employing the beamforming matrix $\mathbf{P}(P)$.
Similarly, we have the set of group-rate tuples $\left\{ \left( r_{1}(P),\ldots,r_{M}(P) \right) \right\}_{P}$, where
$\left( r_{1}(P),\ldots,r_{M}(P) \right)$ is associated with $\mathbf{P}(P)$ and
$\left( R_{1}(P),\ldots,R_{K}(P) \right)$.
A beamforming scheme is also associated with user and group DoF tuples.
The user-DoF tuple is denoted by $\left(D_{1},\ldots,D_{K} \right)$, where the $k$th user-DoF is given as
\begin{equation}
\label{Eq_user_DoF}
D_{k} \triangleq \lim_{P \rightarrow \infty} \frac{R_{k}(P) }{\log_{2}(P)}.
\end{equation}
The corresponding group-DoF tuple is denoted by $\left( d_{1},\ldots,d_{M} \right)$, where $d_{m}$ is given as\footnote{In \eqref{Eq_group_DoF}, (a) follows from continuity by passing the limit in \eqref{Eq_user_DoF} inside the $\min$ function in \eqref{Eq_group_rate}.}
\begin{equation}
\label{Eq_group_DoF}
d_{m} \triangleq \lim_{P \rightarrow \infty} \frac{r_{m}(P) }{\log_{2}(P)} \overset{\text{(a)}}{=} \min_{i \in \mathcal{G}_{m}} D_{i}.
\end{equation}
%%%
It can be seen from \eqref{Eq_user_DoF} and \eqref{Eq_group_DoF} that the DoF metric is independent of the actual power level $P$.
Alternatively, it captures the asymptotic rate growth with respect to the capacity of an interference-free single-stream transmission approximated by $\log_{2}(P)$.
Hence, the tuple $\left( d_{1},\ldots,d_{M} \right)$ can be interpreted as the fractions of interference-free transmissions that can be simultaneously achieved by the $M$ groups at high SNR where inter-group interference is the main limiting factor.
%%%

%%%
A number of meaningful scalar performance measures can be derived from DoF tuples.
This work is concerned with the symmetric performance, and hence we focus on the MMF-DoF (symmetric-DoF)  defined as $d \triangleq \min_{m \in \mathcal{M}}d_{m}$.
%%%
For a given beamforming scheme, $d$ corresponds to the maximum group-DoF that can be simultaneously achieved by all groups.
%%
%%%%
\subsection{MMF-DoF Performance}
\label{Subsection_MMF_DoF_NoRS}
%%%%
%%
For a given system, we denote the optimum MMF beamforming scheme by $\left\{ \mathbf{P}^{\star} (P) \right\}_{P}$.
This is obtained by solving \eqref{Eq_Problem_R_NoRS} for every power level $P$, i.e. $\mathbf{P}^{\star} (P) = \arg \mathcal{R}(P)$.
The optimum MMF beamforming scheme comes with a corresponding MMF-DoF given by $d^{\star} \triangleq \lim_{P \rightarrow \infty}  \frac{\mathcal{R}(P)}{\log_{2}(P)}$, which we characterize in this part.
The fact that each user is equipped with a single antenna sets a trivial upperbound on the MMF-DoF.
In particular, we have
\begin{equation}
\label{Eq_MMF_DoF_upperbound_multicast}
d \leq d_{m} \leq D_{i} \leq 1, \ \forall i \in \mathcal{G}_{m}, m \in \mathcal{M}.
\end{equation}
Hence, if $d= 1$  is achievable, then $d^{\star} = 1$. In this case, it is possible to beam an interference-free stream to each group simultaneously.
For  DoF analysis, we make the following assumption.
%%%%%%%%%%%%%%%%%%%%%%%%%%%%%%%%%%%%%%%%%
\begin{assumption}
\label{Assumption_Generic_Channel}
\textnormal{
The channel vectors $\mathbf{h}_{1},\ldots,\mathbf{h}_{K}$ are independently drawn from continuous distributions.
Hence, for any $N\times K_{\mathrm{sub}}$ matrix in which the $K_{\mathrm{sub}}$ columns constitute any subset of the $K$ channel vectors, it holds with probability one that the rank is $\min\{ N,K_{\mathrm{sub}} \}$.
}
\end{assumption}
%%%%%%%%%%%%%%%%%%%%%%%%%%%%%%%%%%%%%%%%%
%%
Now let us define $\mathbf{H}_{m}$ as the matrix with columns constituting channel vectors of all users in $\mathcal{G}_{m}$.
On the other hand, $\bar{\mathbf{H}}_{m} \triangleq [\mathbf{H}_{1}\cdots\mathbf{H}_{m-1} \; \mathbf{H}_{m+1} \cdots \mathbf{H}_{M}]$ is the matrix composed of the complementary set of channel vectors.
More generally, let $\mathcal{L} = \{m_{1},\ldots,m_{L}\} \subseteq \mathcal{M} $ be a subset of $L$ groups, and
$\bar{\mathcal{L}} = \mathcal{M} \setminus \mathcal{L}$ be the subset of complementary groups.
We define $\mathbf{H}_{\mathcal{L}} \triangleq [\mathbf{H}_{m_{1}} \; \cdots \; \mathbf{H}_{m_{L}}]$ as the channel matrix for all users in groups $\mathcal{G}_{m_{1}}\ldots,\mathcal{G}_{m_{L}}$, and $\bar{\mathbf{H}}_{\mathcal{L}} \triangleq \mathbf{H}_{\bar{\mathcal{L}}}$ as the complementary channel matrix.
From Assumption \ref{Assumption_Generic_Channel}, the number of spatial signalling dimensions orthogonal to the subspace occupied
by receivers in $\mathcal{G}_{m_{1}}\ldots,\mathcal{G}_{m_{L}}$ is given by
\begin{equation}
\label{Eq_null_dimension}
\mathrm{dim} \Big( \mathrm{null}\big( \mathbf{H}_{\mathcal{L}}^{\Hrm}  \big) \Big) =
\max \left\{N-\sum_{l = 1}^{L}G_{m_{l}},0 \right\}.
\end{equation}
It follows that $\mathrm{dim} \left( \mathrm{null}\left( \bar{\mathbf{H}}_{m}^{\Hrm}  \right) \right) \geq 1$ if and only if
\begin{equation}
\label{Eq_condition_N_N_m}
N \geq 1 + K - G_{m}.
\end{equation}
This allows nulling all interference cause by the $m$th beamformer, i.e. $\mathbf{p}_{m} \in \mathrm{null}\big( \bar{\mathbf{H}}_{m}^{\Hrm}  \big)$.
From \eqref{Eq_Group_order}, we have
$\mathrm{dim} \left( \mathrm{null}\left( \bar{\mathbf{H}}_{m_{1}}^{\Hrm}  \right) \right) \leq \mathrm{dim} \left( \mathrm{null}\left( \bar{\mathbf{H}}_{m_{2}}^{\Hrm}  \right) \right)$
for all $m_{1},m_{2} \in \mathcal{M}$  and $m_{1} \leq m_{2}$.
This reflects the fact that nulling interference caused to larger groups is more \emph{demanding} in terms of spatial dimensions.
Hence, satisfying \eqref{Eq_condition_N_N_m} for all $m \in \mathcal{M}$ is equivalent to
\begin{equation}
\label{Eq_min_Nt_max_DoF}
N \geq 1 + K - G_{1}.
\end{equation}
Note that in the extreme case of single-group multicasting $(M = 1)$, the condition in \eqref{Eq_min_Nt_max_DoF}
becomes $N \geq 1$. In this case, a transmitting antenna array provides beamforming gain but no DoF gain as a single antenna achieves the single-stream's DoF upperbound in \eqref{Eq_MMF_DoF_upperbound_multicast}.
On the other hand, the condition in \eqref{Eq_min_Nt_max_DoF} becomes $N \geq K$ for the opposite extreme of multiuser beamforming ($M = K$), under which one spatial dimension per user is necessary to guarantee perfect interference nulling.
By fixing the number of groups and increasing the number of users per group beyond $1$, the dimension of the subspace occupied by each group increases. Hence, interference nulling requires more spatial dimensions.
%%%
A question that comes to mind at this point is: what happens to the MMF-DoF when condition $\eqref{Eq_min_Nt_max_DoF}$ is violated?
Before proceeding, we define
\begin{equation}
\label{Eq_N_L}
N_{L} \triangleq 1 + \sum_{m=2}^{L} G_{m} = 1 + K - G_{1} - \sum_{j=L+1}^{M} G_{j}
\end{equation}
for all $L \in \mathcal{M}$, where $N_{1} = 1$.
$N_{L}$ is interpreted as the minimum number of transmitting antennas required to serve the subset of groups $\{1,\ldots,L\}$ using interference-free beamforming while disregarding all remaining groups.
\begin{proposition}
\label{Proposition_DoF_NoRS}
\textnormal{
%%%
The MMF-DoF achieved by classical beamforming is given by
\begin{align}
\label{Eq_max_min_DoF_NoRS}
d^{\star} \triangleq \lim_{P \rightarrow \infty}  \frac{\mathcal{R}(P)}{\log_{2}(P)} =
\begin{cases}
       1, &  N \geq  N_{M} \\
       0.5, & N_{M-1} + G_{1} \leq N < N_{M} \\
       0, &   N < N_{M-1} + G_{1}.
\end{cases}
\end{align}
}
\end{proposition}
%%%
%%
Results as the one above are commonly shown through two steps: 1) achievability and 2) converse.
In the achievability, it is shown that there exists at least one feasible beamforming scheme that achieves the DoF in \eqref{Eq_max_min_DoF_NoRS}.
In the converse, it is shown that no feasible beamforming scheme can achieve a higher DoF by deriving a tight upperbound.
The achievability of Proposition \ref{Proposition_DoF_NoRS} is discussed in the following, while the converse is relegated to Appendix \ref{Appendix_Converse_NoRS}.
\subsubsection{Achievability of Proposition \ref{Proposition_DoF_NoRS}}
\label{Subsubsection_achievability_proposition_NoRS}
Showing that $d = 1$ under $N \geq  N_{M}$ follows from the discussion that precedes Proposition \ref{Proposition_DoF_NoRS}, and achieving $d = 0$ is trivial.
Hence, we focus on achieving $d = 0.5$.
It is sufficient to show this under $N  =  N_{M-1} + G_{1}$
as further increasing the number of transmitting antennas cannot decrease the DoF.
Next, we observe the following:
\begin{itemize}
\item By excluding the largest group, interference free transmission amongst the remaining $M-1$ groups is possible. This follows by removing the $M$th group from the system and rewriting the condition in \eqref{Eq_min_Nt_max_DoF} as $N \geq  N_{M-1}$.
\item Interference from the $M$th beam to all other groups can be nulled. This follows from \eqref{Eq_condition_N_N_m}.
\end{itemize}
Now consider a beamforming scheme in which a beamforming matrix takes the form
\begin{equation}
\label{Eq_P_NoRS_achievability}
\mathbf{P}(P) = \left[\sqrt{q_{1}(P)}\mathbf{w}_{1} \; \cdots \; \sqrt{q_{M}(P)}\mathbf{w}_{M}\right]
\end{equation}
where the power allocation depends on $P$ while the beamforming directions do not.
The beamforming directions are designed according to the two observations above such that
\begin{align}
\label{Eq_p_m_DoF_05}
\mathbf{w}_{m} \in
\begin{cases}
\mathrm{null}\big( \bar{\mathbf{H}}_{\{m,M\}}^{\Hrm}  \big), & \forall  m \in \mathcal{M} \setminus M\\
\mathrm{null}\big( \bar{\mathbf{H}}_{M}^{\Hrm} \big), & m=M.
\end{cases}
\end{align}
The $k$th user's SINR at power level $P$ is given by
\begin{align}
\gamma_{k}(P) =
\begin{cases}
\frac{q_{\mu(k)}(P) |\mathbf{h}_{k}^{\Hrm} \mathbf{w}_{\mu(k)} |^{2} }{\sigma_{\mathrm{n}}^{2}}, & \forall  k \in \mathcal{K}\setminus \mathcal{G}_{M}\\
 \frac{q_{M}(P)|\mathbf{h}_{k}^{\Hrm} \mathbf{w}_{M} |^{2}}{\sum_{j \neq M} q_{j}(P) |\mathbf{h}_{k}^{\Hrm} \mathbf{w}_{j} |^{2} + \sigma_{\mathrm{n}}^{2}}
, &  \forall k \in \mathcal{G}_{M}.
\end{cases}
\end{align}
It can be seen that users in groups $1,\ldots,M-1$ see no interference at all, while users in group $M$ see interference from all other groups.
Next, power allocation is carried out such that all user SINRs achieve the same power scaling.
This is achieved by power allocations scaling as\footnote{We use the standard Landau notation $O(\cdot)$ to describe power scaling.
Specifically, for real-valued functions $f(P),g(P)$, the statement $f(P) = O\left( g(P) \right)$ means that
$\lim_{P \rightarrow \infty} \frac{|f(P)|}{|g(P)|} < \infty$.}
\begin{align}
\label{Eq_power_scaling_DoF_05}
q_{m}(P) =
\begin{cases}
O\big(P^{0.5} \big), & \forall  m \in \mathcal{M} \setminus M\\
O\big(P \big), & m=M.
\end{cases}
\end{align}
For example, one power allocation that satisfies \eqref{Eq_power_scaling_DoF_05} while adhering to the power constraint is
\begin{align}
q_{m}(P) =
\begin{cases}
\frac{P^{0.5}}{M-1}, & \forall  m \in \mathcal{M} \setminus M\\
P - P^{0.5}, & m=M.
\end{cases}
\end{align}
Since $|\mathbf{h}_{k}^{\Hrm} \mathbf{w}_{m} |^{2} = O(1)$ for all $k \in \mathcal{K}$ and $m \in \mathcal{M}$,
we have $\gamma_{k} = O\big(P^{0.5} \big)$ and $R_{k} = 0.5\log_{2}(P) + O(1)$ for all $k \in \mathcal{K}$.
Hence, the proposed scheme achieves $D_{k} = 0.5$ for all $k \in \mathcal{K}$, from which the group-DoF tuple $(d_{1},\ldots,d_{M}) = (0.5,\ldots,0.5)$ is achieved, and hence $d = 0.5$.

Note that for the DoF achievability, it is sufficient to use simple zero-forcing precoders which are generally suboptimal from a rate perspective. This is a widely observed phenomenon in the MIMO literature, and is due to the fact that the DoF capture the number of interference free dimensions and, unlike achievable rates, are not influenced by $O(1)$ power gains.
\subsubsection{Insight}
It is evident that nulling all interference seen by receivers in $\mathcal{G}_{M}$ is most expensive in terms of spatial dimensions. Alternatively, the scheme reserves the spatial dimensions to achieve interference free transmission amongst the remaining $M-1$ groups.
This comes at the expense of sacrificing part of the $M$th group's received signal subspace, now occupied by interference from the other $M - 1$ beams. Interference is made to scale as $O\big(P^{0.5} \big)$ through power control, which in turn limits the MMF-DoF to $0.5$.
This is shown to be the optimum classical beamforming strategy in the DoF sense in Appendix \ref{Appendix_Converse_NoRS}.
Since multiplexing gains are partially achieved in such scenarios, they are labeled as partially-overloaded.
When $N$ drops below $N_{M-1}+G_{1}$, the interference from the $M$th group cannot be eliminated anymore, creating mutual interference between at least two groups from which one group's gain becomes the other's loss. As a result, the MMF-DoF collapses to zero as shown in Appendix \ref{Appendix_Converse_NoRS}. Such scenarios are identified as fully-overloaded\footnote{For equal size groups, $N_{M} = N_{M-1}+G_{1}$ and the MMF-DoF collapses to $0$ as soon as the condition in \eqref{Eq_min_Nt_max_DoF} is violated.}.
\subsection{The Role of DoF Analysis}
\begin{figure}%[pbth]
\centering
\includegraphics[width = 0.45\textwidth]{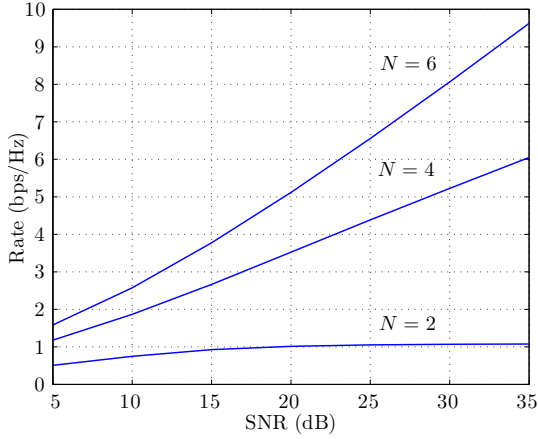}
\caption{MMF rate performance obtained by solving \eqref{Eq_Problem_R_NoRS} using the method in \cite{Karipidis2008}
 for $K=6$ users, $M=3$ groups, $G_{1},G_{2}$ and $G_{3}$ equal to $1,2$ and $3$ users respectively, and
 $N = 2,4$ and $6$ antennas.
 The MMF rates presented correspond to the SDR upperbound (no randomization), and results are obtained by averaging over $100$ i.i.d Rayleigh fading channels.}
\label{Fig_NoRS_example}
\end{figure}
%%%
While a beamforming scheme obtained from \eqref{Eq_Problem_R_NoRS} is guaranteed to achieve the MMF-DoF in \eqref{Eq_max_min_DoF_NoRS}, the converse is not always true, i.e. a scheme that achieves \eqref{Eq_max_min_DoF_NoRS} is not necessarily optimum from the MMF-rate perspective.
For example, while all beamforming schemes that satisfy  \eqref{Eq_p_m_DoF_05} and \eqref{Eq_power_scaling_DoF_05} achieve the same MMF-DoF, the actual beamforming directions and power allocation may have a significant influence on the achievable rate performance.
Moreover, having fixed beamforming directions for all SNR levels is suboptimal in general  as there is more to designing an optimum beamforming scheme than simply nulling interference, which may not be the primary limiting factor in medium and low SNR regimes.
This indeed may raise some questions regarding the role and effectiveness of DoF analysis in the design and optimization process.
To highlight the potential role of DoF analysis in guiding the design of new beamforming strategies, we present a numerical example in
Fig. \ref{Fig_NoRS_example}.
Despite the fact that DoF analysis is carried out as SNR goes to infinity, its results are highly influential and visible at finite SNRs.
For example, the three MMF-DoF regimes characterized in Proposition \ref{Proposition_DoF_NoRS} can be clearly identified from Fig. \ref{Fig_NoRS_example}.
This is due to the dominating effect of intergroup-interference compared to the effect of additive noise.
Moreover, the detrimental implications of $d = 0$ can be clearly observed in Fig. \ref{Fig_NoRS_example}, i.e.
the MMF-rate stops growing as SNR increases, reaching a saturated performance.
%\footnote{To be more precise, the MMF-rate scales as $o\big(\log_{2}(P)\big)$. This implies that it either stops growing or grows extremely slow compared to the interference free scenario, reaching a flat or almost-flat performance.}.
%%
Such fully-overloaded scenarios are characterized through the MMF-DoF as seen in \eqref{Eq_max_min_DoF_NoRS}.
Hence, although the specific finite-SNR rates cannot be precisely predicted from DoF analysis,
insights into the interference-dominated MMF-rate performance can be drawn, guiding the development of more efficient beamforming strategies
as we see in the following sections.
%%%
\section{Degraded Beamforming}
\label{Section_Degraded_Beamforming}
%%%
As we observed in the previous section, the antenna regime in which a system becomes overloaded was identified through DoF analysis.
In this section, we propose a scheme that improves the MMF-DoF in Proposition \ref{Proposition_DoF_NoRS} when $N < N_{M-1}+G_{1}$.
This is achieved through degraded beamforming, where data streams are decoded and cancelled in a successive manner.
The terminology follows from the fact that such (enforced) ordering and successive decoding degrade the channel, hence annihilated all spatial multiplexing gains.
First, to define a decoding order, we use the premutation function $\pi: \mathcal{M} \rightarrow \mathcal{M}$ which permutes the set
$\mathcal{M}$ such that $\pi(m) \neq \pi(j)$ for all $m \neq j$.
An arbitrary user in the $\pi(m)$th group starts by decoding the $\pi(1)$th stream, which is then removed from the received signal using interference cancellation.
This is followed by decoding and removing the $\pi(2)$th stream and so on until the $\pi(m)$th stream is decoded.
Hence, the $\pi(m)$th stream sees interference from the $\pi(j)$th stream only if $j > m$.
Since users may decode streams not intended to them, it is necessary to define the SINR of the $\pi(m)$th stream from the $k$th receiver's perspective as
\begin{equation}
\label{Eq_SINR_Degraded}
\widetilde{\gamma}_{k}^{\pi(m)} =  \frac{|\mathbf{h}_{k}^{\Hrm} \widetilde{\mathbf{p}}_{\pi(m)} |^{2}}{\sum_{j = m+1}^{M}  |\mathbf{h}_{k}^{\Hrm} \widetilde{\mathbf{p}}_{\pi(j)} |^{2} + \sigma_{\mathrm{n}}^{2}}
\end{equation}
where $k$ is not necessarily in $\mathcal{G}_{\pi(m)}$.
In particular, this is relevant for $k \in \{ \mathcal{G}_{\pi(m)},\ldots,\mathcal{G}_{\pi(M)}\}$,
as the remaining users would not reach the $\pi(m)$th stream in the successive decoding chain.
The notation $\widetilde{\cdot}$ in \eqref{Eq_SINR_Degraded} is used to define quantities associated with the degraded transmission.
The rate at which the $\pi(m)$th stream should be transmitted such that the $k$th user is able to  successfully decode it is given by $\widetilde{R}_{k}^{\pi(m)} = \log_{2}(1+\widetilde{\gamma}_{k}^{\pi(m)})$.
Since the $\pi(m)$th stream should be successfully decoded by all receivers
in $\mathcal{G}_{\pi(m)},\ldots,\mathcal{G}_{\pi(M)}$, the $\pi(m)$th group-rate is restricted to
\begin{equation}
\label{Eq_group_rate_Degraded}
r_{\pi(m)} = \min_{i \in \{\mathcal{G}_{\pi(m)},\ldots,\mathcal{G}_{\pi(M)}\}}  \widetilde{R}_{i}^{\pi(m)}.
\end{equation}
The corresponding MMF optimization problem is given by
%%%%
\begin{equation}
\label{Eq_Problem_R_Degraded}
\widetilde{\mathcal{R}}(P):
\begin{cases}
       \underset{\pi,\widetilde{\mathbf{P}}}{\max} & \underset{\pi(m) \in \mathcal{M}}{\min} \
       \underset{i \in \{\mathcal{G}_{\pi(m)},\ldots,\mathcal{G}_{\pi(M)}\}}{\min}  \widetilde{R}_{i}^{\pi(m)}  \\
       \text{s.t.}  & \displaystyle{\sum_{m=1}^{M} \|\widetilde{\mathbf{p}}_{\pi(m)}\|^{2} \leq P}
\end{cases}
\end{equation}
%%%%
where $\widetilde{\mathbf{P}} \triangleq [\widetilde{\mathbf{p}}_{1} \; \cdots \; \widetilde{\mathbf{p}}_{M}]$ is the degraded beamforming matrix.
Note that the permutation function is an optimization variable as the achievable rates are influenced by the decoding order.
\subsection{DoF Analysis}
It is evident that each receiver in $\mathcal{G}_{\pi(M)}$ ends up decoding all $M$ data streams.
This degrades the channel and limits the sum group-DoF to one, which can be further split equally amongst groups.
This is formally shown in the following result.
\begin{proposition}
\label{Proposition_DoF_Degraded}
\textnormal{
%%%
The MMF-DoF achieved by the degraded beamforming strategy is given by
\begin{equation}
\label{Eq_max_min_DoF_Degraded}
\widetilde{d}^{\star} \triangleq
\lim_{P \rightarrow \infty}  \frac{\widetilde{\mathcal{R}}(P)}{\log_{2}(P)} = \frac{1}{M}.
\end{equation}
}
\end{proposition}
%%%
\subsubsection{Proof of Proposition \ref{Proposition_DoF_Degraded}}
%%%
Consider a permutation function given by $\pi(m) = m$.
Next, consider a beamforming scheme where all directions $\widetilde{\mathbf{w}}_{1},\ldots,\widetilde{\mathbf{w}}_{M}$ are randomly chosen
from the space spanned by $\mathbf{H}$ and fixed, while the power allocation is set such that it satisfies the following scaling law
$q_{m}(P) = O\big(P^{(1+M-m)/M} \big)$ for all $m \in \mathcal{M}$.
The interference seen by the $m$th stream is dominated by the $m+1$th stream, which has the next highest power level amongst the remaining interferers.
By substituting the beamforming scheme into \eqref{Eq_SINR_Degraded}, it follows that  $\widetilde{\gamma}_{k}^{m}(P) = O\big(P^{1/M} \big)$.
Combining this with \eqref{Eq_group_rate_Degraded}, the MMF-DoF of $1/M$ is achieved.
%%%

%%%
For the upperbound, consider one user in the $\pi(M)$th group.
Such user decodes all $M$ streams, and the model reduces to a Multiple Access Channel (MAC) with a single-antenna receiver,
which has a sum-DoF of 1.
In particular, for any feasible solution  and $k \in \mathcal{G}_{\pi(M)}$, we observe that
\begin{subequations}
\label{Eq_converse_Degraded}
\begin{align}
\label{Eq_converse_Degraded_1}
 \underset{\pi(m) \in \mathcal{M}}{\min} r_{\pi(m)} & \leq \frac{1}{M}\sum_{m=1}^{M} r_{\pi(m)} \\
\label{Eq_converse_Degraded_2}
 & \leq  \frac{1}{M}\sum_{m=1}^{M} \widetilde{R}_{k}^{\pi(m)}  \\
\label{Eq_converse_Degraded_3}
 & = \frac{1}{M}\log_{2}\left( 1 + \frac{\sum_{m=1}^{M} |\mathbf{h}_{k}^{\Hrm} \widetilde{\mathbf{p}}_{\pi(m)} |^{2} }{\sigma_{\mathrm{n}}^{2}} \right) \\
\label{Eq_converse_Degraded_4}
 & \leq \frac{1}{M}\log_{2}\left( 1 + \frac{ P \|\mathbf{h}_{k}\|^{2}  }{\sigma_{\mathrm{n}}^{2}} \right)
\end{align}
\end{subequations}
where \eqref{Eq_converse_Degraded_1} follows from the fact that the minimum is upperbounded by the average, \eqref{Eq_converse_Degraded_2} follows from $k \in \mathcal{G}_{\pi(M)}$ and \eqref{Eq_group_rate_Degraded}, \eqref{Eq_converse_Degraded_3}
is obtained by substituting  \eqref{Eq_SINR_Degraded} into \eqref{Eq_converse_Degraded_2}, and \eqref{Eq_converse_Degraded_4} follows by applying the Cauchy-Schwarz inequality.
%%%
The upperbound in \eqref{Eq_converse_Degraded_4} scales as $\frac{1}{M}\log_{2}(P) + O(1)$.
%%%
\subsubsection{Insight}
%%%
By forcing a group of users to decode all messages, the achievable DoF performance is similar to that of a degraded channel with only one transmitting antenna\footnote{Note that beamforming gain is achieved by exploiting the multiple antennas. This may improve the achievable rate performance compared to a single-antenna transmitter, but no DoF gain is achieved.}.
Such transmission is able to split the single DoF amongst groups in fully-overloaded scenarios, hence avoiding the collapsing MMF-DoF of classical beamforming.
However, all spatial multiplexing gains achieved by classical beamforming when $N \geq N_{M-1}+G_{1}$ are sacrificed.
As we see in the following section, a better strategy is one that combines the benefits of both beamforming strategies.
Before we proceed, we propose a simplified degraded beamforming strategy.
%%%
\subsection{Single-Stream Degraded Beamforming}
%%%
The optimization problem in \eqref{Eq_Problem_R_Degraded} can be solved by finding the optimum beamformers for each possible decoding order, from which the optimum ordering and beamforming can be obtained.
However, finding the optimum multicasting beamformers (even for fixed ordering) is known to be a very difficult task, making the problem highly complicated.
Alternatively, the problem is made easier by imposing a level of suboptimality.
In particular, by restricting all degraded beamforming directions to a common direction $\mathbf{w}_{\mathrm{c}}$, the transmitted signal is expressed by
\begin{equation}
\label{Eq_x_Degraded_simplified}
\mathbf{x} = \mathbf{w}_{\mathrm{c}} \sum_{m=1}^{M} \sqrt{q_{m}} s_{m}
= \mathbf{p}_{\mathrm{c}} \sum_{m=1}^{M} \sqrt{\frac{q_{m}}{P}} s_{m}
\end{equation}
where $\mathbf{p}_{\mathrm{c}} = \sqrt{P}\mathbf{w}_{\mathrm{c}}$.
We further assume that each of the $K$ receivers decodes all $M$ streams.
This imposes the MAC upperbound in \eqref{Eq_converse_Degraded} on all receivers,
from which the upperbound on the MMF group-rate is tightened such that
\begin{equation}
\label{Eq_Degraded_simplified_upperbound}
\min_{\pi(m) \in \mathcal{M}} r_{\pi(m)}  \leq \frac{1}{M} \min_{k \in \mathcal{K}} \log_{2}\left( 1 + \frac{P|\mathbf{h}_{k}^{\Hrm} \mathbf{w}_{\mathrm{c}} |^{2} }{\sigma_{\mathrm{n}}^{2}} \right).
\end{equation}
It can be shown that the upperbound in
\eqref{Eq_Degraded_simplified_upperbound} is attainable as
the superposition of the $M$ streams $\sum_{m=1}^{M} \sqrt{\frac{q_{m}}{P}} s_{m}$ can be equivalently replaced by one super symbol stream which carries all $M$ messages, i.e. $W_{1},\ldots,W_{M} \mapsto s_{\mathrm{c}}$. This is transmitted along the beamformer $\mathbf{p}_{\mathrm{c}}$ in a single-group multicasting manner such that all $K$ receivers are able to decode it.
The upperbound in \eqref{Eq_Degraded_simplified_upperbound} corresponds to the equal splitting of the single-stream's rate amongst the $M$ groups.

From the above discussion, it follows that optimizing the simplified degraded beamforming scheme is equivalent to
packing all $K$ users in one group and solving the single-group multicast beamforming problem.
This single-stream problem is expressed by
%%%%
\begin{equation}
\label{Eq_Problem_R_Degraded_SS}
\mathcal{R}_{\mathrm{c}}(P):
\begin{cases}
      \underset{\mathbf{p}_{\mathrm{c}} }{\max} & \frac{1}{M} \underset{k \in \mathcal{K}}{\min}
       \log_{2}\left( 1 + \frac{|\mathbf{h}_{k}^{\Hrm} \mathbf{p}_{\mathrm{c}} |^{2} }{\sigma_{\mathrm{n}}^{2}} \right)  \\
       \text{s.t.}  & \displaystyle{ \|\mathbf{p}_{\mathrm{c}}\|^{2} \leq P}
\end{cases}
\end{equation}
%%%%
which can be solved using existing methods.
The next question that comes to mind is, what is the performance loss from adopting the single-stream strategy in \eqref{Eq_Problem_R_Degraded_SS}?
It can be seen that the single-stream strategy does not concede a DoF loss compared to the degraded beamforming strategy as
$\lim_{P \rightarrow \infty}  \frac{\mathcal{R}_{\mathrm{c}}(P)}{\log_{2}(P)} = \frac{1}{M}$.
Both strategies translate to the equal splitting of a single DoF amongst the $M$ groups.
However, the single DoF is accessed in a non-orthogonal manner in \eqref{Eq_Problem_R_Degraded}, while \eqref{Eq_Problem_R_Degraded_SS} is equivalent to Orthogonal Multiple Access (OMA).
This yields a rate gap at finite SNRs. The analysis of such gap is outside the scope of this work.
In the remainder of the paper, we restrict ourselves to the simplified degraded scheme.
%%
%Such loss can be upperbounded by $\frac{1}{M} \log_{2} \left( \frac{\max_{k \in \mathcal{K}} \| \mathbf{h}_{k} \|^{2} }{\min_{i \in \mathcal{K}} |\mathbf{e}_{1}^{\Trn} \mathbf{h}_{i}|^{2}}  \right)$,
%%%
%which follows by combining the upperbound in \eqref{Eq_converse_Degraded} with the rate achieved by single-antenna single-stream transmission.
%%%
%Such gap can be further tightened and analysed for specific distributions and varying strengths of user-channels, which is outside the scope of this paper.
%%%
%In the remainder of the paper, we restrict ourselves to the simplified degraded scheme.
%%%
\section{Rate-Splitted Beamforming}
\label{Section_RS}
%%%
As we saw in the previous sections, both the classical and the degraded strategies have their benefits
and limitations.
The former exploits the multiplexing gains offered by the antenna array yet fails in the overloaded regime,
while the latter guarantees a non-saturating MMF performance yet fails to utilize the DoF gains achieved through
spatial multiplexing. Here, we propose a RS beamforming strategy that is able to reap the fruits of both strategies.
In the RS strategy, each message is split into two parts: degraded and designated.
For example, $W_{m} \mapsto W_{m0},W_{m1} $ where $W_{m0}$ and $W_{m1}$ denote the degraded and designated parts respectively.
Degraded parts are encoded into degraded signals in the manner described in Section \ref{Section_Degraded_Beamforming}, while designated parts are encoded into designated signals as described in Section \ref{Section_NoRS_Problem}.
All signals are superposed and transmitted simultaneously.
To simplify the analysis, design and optimization, single-stream beamforming is employed to construct one degraded signal.
In particular, degraded parts are encoded into one super degraded symbol stream as
$W_{10},\ldots,W_{M0} \mapsto s_{\mathrm{c}}$.
On the other hand, designated parts are encoded into independent symbols streams as
$W_{11},\ldots,W_{M1} \mapsto s_{1},\ldots,s_{M}$.
The transmitted signal is then constructed as
\begin{equation}
\mathbf{x} = \mathbf{p}_{\mathrm{c}} s_{\mathrm{c}} + \sum_{m=1}^{M} \mathbf{p}_{m} s_{m}.
\end{equation}
Since the information intended to the $m$th group is contained partially in $s_{\mathrm{c}}$ and partially in $s_{m}$, both streams should be decoded by all receivers in  $\mathcal{G}_{m}$.
Hence, at the $k$th receiver, the degraded stream is first decoded by treating all designated streams as noise.
This is followed by removing the degraded part from the received signal using interference cancellation, before decoding the designated stream while treating all remaining streams as noise.
The $k$th receiver retrieves its message given that $s_{\mathrm{c}}$ and $s_{\mu(k)}$ are successfully decoded.
\subsection{Problem Formulation}
The degraded stream's SINR at the $k$th user is given by
\begin{equation}
\label{Eq_SINR_c_RS}
\gamma_{k}^{\mathrm{c}} =  \frac{|\mathbf{h}_{k}^{\Hrm} \mathbf{p}_{\mathrm{c}} |^{2}}
{\sum_{j =1}^{M} |\mathbf{h}_{k}^{\Hrm} \mathbf{p}_{j} |^{2} + \sigma_{\mathrm{n}}^{2}}.
\end{equation}
The corresponding achievable rate from the $k$th user's point of view is given by $R_{k}^{\mathrm{c}} = \log_{2}(1+\gamma_{k}^{\mathrm{c}})$.
For all users to be able to successfully decode the degraded stream, its transmission rate should be restricted to
$R^{\mathrm{c}} = \min_{k \in \mathcal{K}} R_{k}^{\mathrm{c}}$.
After cancelling the degraded signal, the SINRs and achievable rates of the designated schemes are as defined in Section \ref{Section_NoRS_Problem}.
By performing RS, each group-rate is composed of a sum of two portions corresponding to the degraded and designated parts of the message.
The degraded portion of the $m$th group-rate, contributed by the degraded stream, is given by
\begin{equation}
C_{m} = \frac{|W_{m0}|}{\sum_{j = 1}^{M}|W_{j0}|}R^{\mathrm{c}}
\end{equation}
where $|W|$ is the length of a message $W$. It naturally follows that $\sum_{m=1}^{M} C_{m} = R^{\mathrm{c}}$.
Note that $C_{m}$ corresponds to the rate at which $W_{m0}$ is communicated.
On the other hand, $W_{m1}$ is communicated at the designated rate given by $\min_{i \in \mathcal{G}_{m}} R_{i}$.
Hence, the $m$th group-rate is given by
\begin{equation}
r_{m} = C_{m} + \min_{i \in \mathcal{G}_{m}} R_{i}.
\end{equation}
This allows us to formulate the RS problem as
%%
%%%%
\begin{equation}
\label{Eq_Problem_R_RS}
\mathcal{R}_{\mathrm{RS}}(P):
\begin{cases}
       \underset{\mathbf{c},\mathbf{P}_{\mathrm{RS}}}{\max} & \underset{m \in \mathcal{M}}{\min}
       \Big( C_{m} + \underset{i \in \mathcal{G}_{m}}{\min} \ R_{i} \Big) \\
       \text{s.t.}
       &  R_{k}^{\mathrm{c}} \geq \sum_{m=1}^{M} C_{m} , \forall k \in \mathcal{K}   \\
       &   C_{m} \geq 0, \forall m \in \mathcal{M} \\
       & \displaystyle{ \|\mathbf{p}_{\mathrm{c}}\|^{2} +  \sum_{m=1}^{M} \|\mathbf{p}_{m}\|^{2} \leq P}
\end{cases}
\end{equation}
%%%%
where $\mathbf{c} \triangleq ( C_{1},\ldots,C_{M} )$ and
$\mathbf{P}_{\mathrm{RS}} \triangleq [\mathbf{p}_{\mathrm{c}} \; \mathbf{p}_{1} \; \cdots \; \mathbf{p}_{M}]$.
The constraint $ R_{k}^{\mathrm{c}} \geq \sum_{m=1}^{M} C_{m}$ guarantees that the degraded stream is decoded by the $k$th user.
A solution to \eqref{Eq_Problem_R_RS} consists of the rates, the splitting ratios which can be deduced from the rates, and the beamforming vectors.
By inspecting the problem formulations in \eqref{Eq_Problem_R_NoRS}, \eqref{Eq_Problem_R_Degraded_SS} and \eqref{Eq_Problem_R_RS}, we obtain the relationship
\begin{equation}
\label{Eq_MMF_Rate_RS_Degraded_NoRS}
\max\big\{\mathcal{R}(P), \mathcal{R}_{\mathrm{c}}(P) \big\}
\leq
\mathcal{R}_{\mathrm{RS}}(P).
\end{equation}
This follows by observing that optimum solutions of problems \eqref{Eq_Problem_R_NoRS} and \eqref{Eq_Problem_R_Degraded_SS} correspond to feasible solutions of problem \eqref{Eq_Problem_R_RS} with $\|\mathbf{p}_{\mathrm{c}}\|^{2} = 0$ and
$\|\mathbf{p}_{1}\|^{2},\ldots,\|\mathbf{p}_{M}\|^{2} = 0$ respectively.
While the inequality \eqref{Eq_MMF_Rate_RS_Degraded_NoRS} confirms that the RS beamforming strategy cannot perform worse than the two preceding strategies, it does not quantify the performance improvement.
This can be partially settled by characterizing the MMF-DoF performance achieved by RS.
To facilitate such derivation, we introduce the following special case.
\subsection{A Special Case: Partitioned Beamforming}
Consider a strategy where groups are partitioned into two subsets, namely
$\mathcal{M}_{\mathrm{D}}\subseteq \mathcal{M}$ which are served using classical designated beamforming,  and
$\mathcal{M}_{\mathrm{c}} = \bar{\mathcal{M}}_{\mathrm{D}} = \mathcal{M} \setminus \mathcal{M}_{\mathrm{D}}$ served through degraded beamforming.
It is clear that this is a special case of the RS strategy achieved by splitting the messages such that $|W_{m0}| = 0$ for all $m \in \mathcal{M}_{\mathrm{D}}$, and $|W_{m1}| = 0$ for all $m \in \mathcal{M}_{\mathrm{c}}$.
This yields a design with $C_{m} = 0$ for all $m \in \mathcal{M}_{\mathrm{D}}$, and $\|\mathbf{p}_{m}\|^{2} = 0$ for all $m \in \mathcal{M}_{\mathrm{c}}$.
The degraded stream's SINR at the $k$th user is now given by
\begin{equation}
\label{Eq_SINR_c_Partitioned}
\gamma_{k}^{\mathrm{c}} =  \frac{|\mathbf{h}_{k}^{\Hrm} \mathbf{p}_{\mathrm{c}} |^{2}}
{\sum_{j \in \mathcal{M}_{\mathrm{D}}} |\mathbf{h}_{k}^{\Hrm} \mathbf{p}_{j} |^{2} + \sigma_{\mathrm{n}}^{2}}
\end{equation}
where the corresponding rate writes as $R_{k}^{\mathrm{c}} = \log_{2}(1+\gamma_{k}^{\mathrm{c}})$.
While $s_{\mathrm{c}}$ only carries messages intended to groups in $\mathcal{M}_{\mathrm{c}}$, it is still decoded by receivers in all groups as in the RS strategy to improve the decodability of designated streams.
Hence, the rate of the degraded stream is given by  $R^{\mathrm{c}} = \min_{k \in \mathcal{K}} R_{k}^{\mathrm{c}} $.
For the $k$th receiver where $k \in \left\{\mathcal{G}_{m} \mid m \in \mathcal{M}_{\mathrm{D}}\right\}$, the SINR of the designated stream is given by
\begin{equation}
\label{Eq_SINR_k_Partitioned}
\gamma_{k} =  \frac{|\mathbf{h}_{k}^{\Hrm} \mathbf{p}_{\mu(k)} |^{2}}
{\sum_{j \in \mathcal{M}_{\mathrm{D}}\setminus \mu(k)} |\mathbf{h}_{k}^{\Hrm} \mathbf{p}_{j} |^{2} + \sigma_{\mathrm{n}}^{2}}
\end{equation}
and the corresponding rate is given by $R_{k} = \log_{2}(1+\gamma_{k})$.
Achieving fairness in this case requires sharing $R^{\mathrm{c}}$ equally amongst groups in $\mathcal{M}_{\mathrm{c}}$.
It follows that the group-rates are given by
\begin{align}
\label{Eq_r_m_Partitioned}
r_{m} =
\begin{cases}
\frac{1}{|\mathcal{M}_{\mathrm{c}}|}\underset{k \in \mathcal{K}}{\min} \; R^{\mathrm{c}}_{k}, & \forall  m \in \mathcal{M}_{\mathrm{c}}\\
\min_{i \in \mathcal{G}_{m}}R_{i}, &  \forall m \in \mathcal{M}_{\mathrm{D}}.
\end{cases}
\end{align}
It is evident from \eqref{Eq_r_m_Partitioned} that the manner in which groups are partitioned has an influence on the achievable rate performance.
This is exploited in the DoF analysis presented next.
\subsection{DoF Analysis}
Now we are ready to derive the following result.
\begin{proposition}
\label{Proposition_DoF_RS}
\textnormal{
The MMF-DoF achieved by the RS beamforming strategy is given by
\begin{equation}
\label{Eq_max_min_DoF_RS}
d_{\mathrm{RS}}^{\star} \triangleq
\lim_{P \rightarrow \infty}  \frac{\mathcal{R}_{\mathrm{RS}}(P)}{\log_{2}(P)}= \frac{1}{1 + M - M_{\mathrm{D}}^{\star}}.
\end{equation}
where
\begin{align}
\label{Eq_M_D_star}
M_{\mathrm{D}}^{\star} \! = \!
\begin{cases}
       M, \;             N        \geq  N_{M} \\
       L, \;             N_{L} \!  \leq  \! N \! < \! N_{L+1}, \; \forall L \! \in \! \{ 1,\ldots,M-1 \}
\end{cases}
\end{align}
and $N_{L}$ is expressed in \eqref{Eq_N_L}.
}
\end{proposition}
The achievability of Proposition \ref{Proposition_DoF_RS} is based on partitioned beamforming as shown in what follows.
The converse on the other hand is relegated to Appendix \ref{Appendix_Converse_RS}.
Before we proceed, we highlight that $M_{\mathrm{D}}^{\star}$ in \eqref{Eq_M_D_star} is in fact
the maximum number of groups that can be served through interference-free designated beamforming (i.e. achieving a group-DoF of $1$ each) while silencing all the remaining groups.
This is shown as follows.
Assume that we wish to serve the subset of groups $\mathcal{L} = \{m_{1},\ldots,m_{L}\}$ using interference-free designated beamforming and disregard the remaining groups.
We further assume, without loss of generality, that $\mathcal{L}$ has an ascending order.
It follows from the discussion in Section \ref{Subsection_MMF_DoF_NoRS} that the minimum number of antennas required to do so is
$N_{\min}(\mathcal{L}) = 1+ \sum_{l = 2}^{L} G_{m_{l}}$.
It can be easily seen that
$N_{L} \leq N_{\min}(\mathcal{L})$  for all $\mathcal{L} \subseteq \mathcal{M} \ \text{and} \ |\mathcal{L}| = L$.
Hence for a fixed $L \in \mathcal{M}$, the subset of groups that requires the \emph{least} number of transmitting antennas is $\{1,\ldots,L\}$, i.e. the $L$ smallest groups.
Conversely, for a given number of antennas $N$, the maximum number of groups that can be served through interference-free designated beamforming is given by $M_{\mathrm{D}}^{\star}$.
This follows from $N_{M_{\mathrm{D}}^{\star}} \leq N < N_{M_{\mathrm{D}}^{\star}+1}$ for all $M_{\mathrm{D}}^{\star} \in \{1,\ldots,M-1\}$ and applying a contradiction argument.
\subsubsection{Achievability of Proposition \ref{Proposition_DoF_RS}}
\label{subsubsection_achievability_RS}
Consider a partitioned beamforming scheme where the subset $\mathcal{M}_{\mathrm{D}} = \{1,\ldots,M_{\mathrm{D}}^{\star}\}$ is served using designated beamforming while the remaining groups are served through degraded beamforming.
We observe that by disregarding groups in $\mathcal{M}_{\mathrm{c}}$, it is possible to carry out interference-free designated beamforming amongst all groups in $\mathcal{M}_{\mathrm{D}}$ which follows from $N \geq N_{M_{\mathrm{D}}^{\star}} $.
Hence, the designated beamforming directions are designed as
\begin{equation}
\mathbf{w}_{m} \in
\mathrm{null}\Big( \bar{\mathbf{H}}_{\{m,\mathcal{M}_{\mathrm{c}}\}}^{\Hrm}  \Big), \;  \forall  m \in \mathcal{M}_{\mathrm{D}}.
\end{equation}
On the other hand, $\mathbf{w}_{\mathrm{c}}$ is chosen randomly.
The power allocation is made to scale as
\begin{align}
q_{m}(P) & = O\big(P^{\alpha} \big), \; \forall  m \in \mathcal{M}_{\mathrm{D}} \\
q_{\mathrm{c}}(P) & = O(P)
\end{align}
where $\alpha \in [0,1]$ is some power partitioning factor.
By decoding the degraded stream while treating all other streams as noise, we have $\gamma_{k}^{\mathrm{c}}(P) = O(P^{1-\alpha})$ for all $k \in \mathcal{K}$, which follows from  \eqref{Eq_SINR_c_Partitioned}.
Hence, the degraded super symbol achieves a DoF of $1-\alpha$.
This DoF is divided equally amongst the degraded groups, from which the group-DoF is given by
\begin{equation}
d_{m} = \frac{1-\alpha}{M - M_{\mathrm{D}}^{\star}}, \; \forall m \in \mathcal{M}_{\mathrm{c}}.
\end{equation}
after removing the degraded stream from the received signal, it can be seen from \eqref{Eq_SINR_k_Partitioned} that
$\gamma_{k}(P) = O(P^{\alpha})$ for all $k \in \{\mathcal{G}_{m} \mid m \in \mathcal{M}_{\mathrm{D}}\}$.
Hence, we have $d_{m} = \alpha$ for all $m \in \mathcal{M}_{\mathrm{D}}$.
By setting $\alpha = \frac{1}{1 + M - M_{\mathrm{D}}^{\star}}$, the MMF-DoF in \eqref{Eq_max_min_DoF_RS} is achieved.
The fact that the proposed partitioned scheme can be realized by a corresponding RS scheme completes the achievability.
The above scheme can be viewed as a signal-space partitioning scheme \cite{Yuan2016}.
In particular, the signal-space is divided such that
the bottom power levels (up to $\alpha$) are reserved for interference-free designated beamforming, while the remaining top power levels (from $\alpha$ to $1$) are used for degraded beamforming.
The degraded beam carries a total DoF of $1-\alpha$ as it sees interference from the bottom $\alpha$ power levels, while each designated beam carries a DoF of $\alpha$.
Since the degraded DoF ($1-\alpha$) gets divided by $|\mathcal{M}_{\mathrm{c}}|$ while the designated DoF
($\alpha$) is multiplied by $|\mathcal{M}_{\mathrm{D}}|$, it is natural to divide groups such that $|\mathcal{M}_{\mathrm{D}}|$ is maximized, which in turn minimizes $|\mathcal{M}_{\mathrm{c}}|$.
This is achieved by the proposed group partitioning.
\subsubsection{Insight}
\label{subsubsection_RS_insight}
To gain insight into the MMF-DoF in Proposition \ref{Proposition_DoF_RS}, the result is described as
\begin{align}
\label{Eq_max_min_DoF_RS_2}
d_{\mathrm{RS}}^{\star} =
\begin{cases}
       1,             & N        \geq  N_{M} \\
       \frac{1}{2},   & N_{M-1}  \leq N < N_{M} \\
       \frac{1}{3},   & N_{M-2}  \leq N < N_{M-1} \\
       \vdots         & \vdots   \\
       \frac{1}{M-1}, & N_{2}    \leq N < N_{3} \\
        \frac{1}{M},  & 1        \leq N < N_{2}
\end{cases}
\end{align}
which is obtained by substituting \eqref{Eq_M_D_star} into \eqref{Eq_max_min_DoF_RS}.
By comparing \eqref{Eq_max_min_DoF_RS_2} to \eqref{Eq_max_min_DoF_NoRS} and \eqref{Eq_max_min_DoF_Degraded}, we see that in addition to combining the advantages of the designated and degraded strategies, the RS strategy surpasses both in some cases.
For example, consider $N = N_{M-1}$.
The first observation in Section \ref{Subsubsection_achievability_proposition_NoRS} holds, whilst the second does not.
Hence, the interference caused by the $M$th stream cannot be nulled through designated beamforming which limits the MMF-DoF to $d^{\star} = 0$.
However, $d_{\mathrm{RS}}^{\star} = 0.5$ is achieved by transmitting the $M$th stream in a degraded manner and removing it from the received signals through interference cancellation.
This is also strictly greater than $\widetilde{d}^{\star} = \frac{1}{M}$ for $M > 2$, due to the multiplexing gain of designated beams.
As we saw from the achievability of Proposition \ref{Proposition_DoF_RS}, the RS strategy's optimum MMF-DoF is achieved through partitioned beamforming where no splitting is necessary.
While this holds in the DoF sense, it is not necessarily the case when considering the achievable rates at finite SNRs.
This is confirmed in the simulation results, where it is shown that message splitting is in fact beneficial at finite SNRs.
From a problem-solving perspective, we observe that the RS formulation in \eqref{Eq_Problem_R_RS} avoids the joint optimization of the beamforming matrix and group assignment in partitioned beamforming.
Alternatively, the beamforming matrix and rate allocations are obtained by solving the problem in \eqref{Eq_Problem_R_RS} as we see in the following section.
\section{Optimization}
\label{Section_Optimization}
While \eqref{Eq_Problem_R_NoRS} and \eqref{Eq_Problem_R_Degraded_SS} can be formulated in terms of the SINRs and solved using existing methods, e.g. \cite{Sidiropoulos2006,Karipidis2008,Schad2012}, this cannot be applied to the RS problem in \eqref{Eq_Problem_R_RS} as the performance of each user cannot be represented by a single SINR expression.
As seen from \eqref{Eq_Problem_R_RS}, each achievable user-rate (and ultimately group-rate) is in fact expressed as a sum-rate.
For this reason, we resort to the WMMSE approach  \cite{Christensen2008,Shi2011}, which is particularly effective in dealing with problems incorporating non-convex coupled sum-rate expressions, including RS problems \cite{Joudeh2016b,Joudeh2016a}.
\subsection{Rate-WMMSE Relationship}
We start by establishing the Rate-WMMSE relationship, around which the WMMSE algorithm is based.
First, let us express the $k$th user's average received power as
\begin{equation}
\label{Eq_T_c_k}
T_{\mathrm{c},k} = |\mathbf{h}_{k}^{\Hrm}\mathbf{p}_{\mathrm{c}}|^{2} +
\underbrace{|\mathbf{h}_{k}^{\Hrm}\mathbf{p}_{\mu(k)}|^{2} + \overbrace{\sum_{m\neq \mu(k)} |\mathbf{h}_{k}^{\Hrm}\mathbf{p}_{m}|^{2} + \sigma_{\mathrm{n}}^{2}}^{I_{k}} }_{T_{k} = I_{\mathrm{c},k}}.
\end{equation}
Next, we define the MSEs.
The $k$th user's estimate of $s_{\mathrm{c}}$, denoted by $\widehat{s}_{\mathrm{c},k}$, is obtained by applying the  equalizer $g_{\mathrm{c},k}$ to the receive signal such that $\widehat{s}_{\mathrm{c},k} = g_{\mathrm{c},k}y_{k}$.
After removing the common stream from the received signal, the equalizer $g_{k}$ is applied to the remaining signal to obtain an estimate of $s_{k}$ given by $\widehat{s}_{k}=g_{k}(y_{k}-\mathbf{h}_{k}^{\Hrm}\mathbf{p}_{\mathrm{c}}s_{\mathrm{c}})$.
The common and private MSEs at the output of the $k$th receiver, defined as $\varepsilon_{\mathrm{c},k} \triangleq \E\{|\widehat{s}_{\mathrm{c},k} - s_{\mathrm{c}}|^{2}\}$ and $\varepsilon_{k} \triangleq \E\{|\widehat{s}_{k} - s_{k}|^{2}\}$ respectively, are given by
\begin{equation}
\label{Eq_MSE}
\begin{split}
  \varepsilon_{\mathrm{c},k}  & = |g_{\mathrm{c},k}|^{2} T_{\mathrm{c},k} -2\Re \big\{g_{\mathrm{c},k}\mathbf{h}_{k}^{\Hrm}\mathbf{p}_{\mathrm{c}}\big\}+1
\quad \text{and} \\
\varepsilon_{k}  & =  |g_{k}|^{2} T_{k}-2\Re \big\{g_{k}\mathbf{h}_{k}^{\Hrm}\mathbf{p}_{\mu}\big\}+1.
\end{split}
\end{equation}
By optimizing the MSEs, the MMSEs are obtained as
$\varepsilon_{\mathrm{c},k}^{\mathrm{MMSE}}   \triangleq \min_{g_{\mathrm{c},k}} \varepsilon_{\mathrm{c},k} =
T_{\mathrm{c},k}^{-1} I_{\mathrm{c},k} $ and
$\varepsilon_{k}^{\mathrm{MMSE}}  \triangleq \min_{g_{k}} \varepsilon_{k} = T_{k}^{-1}I_{k}$
%%%
%%%
where the corresponding optimum equalizers are the well-known MMSE weights given by
$g_{\mathrm{c},k}^{\mathrm{MMSE}} = \mathbf{p}_{\mathrm{c}}^{\Hrm}\mathbf{h}_{k} T_{\mathrm{c},k}^{-1}$
and
$g_{k}^{\mathrm{MMSE}} = \mathbf{p}_{k}^{\Hrm}\mathbf{h}_{k}T_{k}^{-1}$.
The MMSEs are related to the SINRs such that
$\gamma_{k}^{\mathrm{c}} = \big( 1/\varepsilon_{\mathrm{c},k}^{\mathrm{MMSE}} \big) - 1 $
and
$\gamma_{k} = \big( 1/\varepsilon_{k}^{\mathrm{MMSE}} \big) - 1 $,
from which the achievable rates write as
$R_{k}^{\mathrm{c}} = -\log_{2}(\varepsilon_{\mathrm{c},k}^{\mathrm{MMSE}})$ and
$R_{k} = -\log_{2}(\varepsilon_{k}^{\mathrm{MMSE}})$.
Next we introduce the main building blocks of the solution, the augmented WMSEs defined for the $k$th user as:
\begin{equation}
\label{Eq_A_WMSEs}
\xi_{\mathrm{c},k}
 \triangleq
u_{\mathrm{c},k} \varepsilon_{\mathrm{c},k}  -  \log_{2} (  u_{\mathrm{c},k} )
\ \ \ \text{and} \ \ \
\xi_{k}
 \triangleq
u_{k} \varepsilon_{k}  -  \log_{2}  (  u_{k} )
\end{equation}
where $u_{\mathrm{c},k}, u_{k} > 0$ are the corresponding weights.
In the following, $\xi_{\mathrm{c},k}$ and $\xi_{k}$ are referred to as the WMSEs for brevity.
The Rate-WMMSE relationship is established by optimizing \eqref{Eq_A_WMSEs} over the equalizers and weights such that:
\begin{equation}
\label{Eq_min_WMSE}
\begin{split}
 \xi_{\mathrm{c},k}^{\mathrm{MMSE}}  & \triangleq  \underset{u_{\mathrm{c},k}, g_{\mathrm{c},k}}{\min} \xi_{\mathrm{c},k} = 1-R_{k}^{\mathrm{c}}
\quad \text{and} \\
 \xi_{k}^{\mathrm{MMSE}} &  \triangleq \underset{u_{k}, g_{k}}{\min} \ \xi_{k}= 1-R_{k}
\end{split}
\end{equation}
where the optimum equalizers are given by: $g_{\mathrm{c},k}^{\star}  = g_{\mathrm{c},k}^{\mathrm{MMSE}} $ and $g_{k}^{\star} = g_{k}^{\mathrm{MMSE}}$, and the optimum weights are given by:
$u_{\mathrm{c},k}^{\star} = u_{\mathrm{c},k}^{\mathrm{MMSE}} \triangleq \big( \varepsilon_{\mathrm{c},k}^{\mathrm{MMSE}} \big)^{-1}$
and
$u_{k}^{\star} = u_{k}^{\mathrm{MMSE}} \triangleq \big( \varepsilon_{k}^{\mathrm{MMSE}} \big)^{-1}$.
This is obtained by checking the first order optimality conditions.
By closely examining each WMSE, it can be seen that it is convex in each variable while fixing the other two.
\subsection{WMSE Reformulation and Algorithm}
Motivated by \eqref{Eq_min_WMSE}, an equivalent WMSE reformulation of problem \eqref{Eq_Problem_R_RS} writes as
%%%%
\begin{equation}
\label{Eq_Problem_WMSE_RS}
\widehat{\mathcal{R}}_{\mathrm{RS}}(P):
\begin{cases}
       \underset{r_{\mathrm{g}},\mathbf{r}, \mathbf{c},\mathbf{P}_{\mathrm{RS}},\mathbf{g},\mathbf{u}}{\max} \ \  r_{\mathrm{g}} \\
        \text{s.t.} \quad
         C_{m} + r_{m} \geq  r_{\mathrm{g}}, \forall m \in \mathcal{M}  \\
        \quad \quad 1-\xi_{i} \geq r_{m}, \forall i \in \mathcal{G}_{m}, \forall m \in \mathcal{M} \\
        \quad \quad 1-\xi_{\mathrm{c},k} \geq \sum_{m=1}^{M} C_{m} , \forall k \in \mathcal{K}   \\
        \quad \quad  C_{m} \geq 0, \forall m \in \mathcal{M} \\
        \quad \quad \displaystyle{ \|\mathbf{p}_{\mathrm{c}}\|^{2} +  \sum_{m=1}^{M} \|\mathbf{p}_{m}\|^{2} \leq P}
\end{cases}
\end{equation}
%%%%
where $r_{\mathrm{g}}$ and $\mathbf{r} \triangleq (r_{1},\ldots,r_{M})$  are auxiliary variables,
$\mathbf{u} \triangleq \left( u_{\mathrm{c},k},u_{k} \mid k \in \mathcal{K} \right)$ is is the set of weights, and $\mathbf{g} \triangleq \left( g_{\mathrm{c},k},g_{k}\mid k \in \mathcal{K} \right)$ is the set of equalizers.
%%
%The equivalence between \eqref{Eq_Problem_WMSE_RS} and \eqref{Eq_Problem_R_RS} is established by observing that the WMSEs are decoupled in their equalizers and weights. Hence, for a given $\mathbf{P}_{\mathrm{RS}}$, optimum $\mathbf{g}$ and $\mathbf{u}$ are obtained by minimizing each WMSE separately, yielding the MMSE solution in \eqref{Eq_min_WMSE}.
%%%
%The equivalence follows by substituting \eqref{Eq_min_WMSE} into \eqref{Eq_Problem_WMSE_RS}.
%%
The WMSE problem in \eqref{Eq_Problem_WMSE_RS} is solved using an Alternating Optimization (AO) algorithm.
In a given iteration of the algorithm, $\mathbf{g}$ and $\mathbf{u}$ are firstly updated using the optimum MMSE solution of \eqref{Eq_min_WMSE}.
Next, $\mathbf{P}_{\mathrm{RS}}$ alongside all auxiliary variables in \eqref{Eq_Problem_WMSE_RS}
are updated by solving \eqref{Eq_Problem_WMSE_RS} for fixed $\mathbf{g}$ and $\mathbf{u}$.
This is a convex optimization problem which can be solved using interior-point methods \cite{Boyd2004}.
%%
%The steps of the AO procedure are summarized in Algorithm \ref{Algthm_AO}.
%%%
%%%%%%%%%%%%%%%%%%%%%%%%%%
%\begin{algorithm}%[h]
%\caption{The WMMSE approach}
%\label{Algthm_AO}
%\begin{algorithmic}[1]
%\State Initialize: $n\gets 0$, $\mathbf{P}_{\mathrm{RS}}$, $r_{\mathrm{g}}^{(n)} \gets 0$
%%
%%
%\Repeat
%    \State $n\gets n+1$
%    \State $\big(g_{\mathrm{c},k},g_{k}\big)
%    \gets \big(g_{\mathrm{c},k}^{\mathrm{MMSE}},g_{k}^{\mathrm{MMSE}}\big)$,
%    $\forall k \in \mathcal{K}$
%    %%%%%%
%    \State $\big(u_{\mathrm{c},k},u_{k}\big)
%    \gets \big(u_{\mathrm{c},k}^{\mathrm{MMSE}},u_{k}^{\mathrm{MMSE}}\big)$,
%    $\forall k \in \mathcal{K}$
%    %
%    \State $(r_{\mathrm{g}}^{(n)},\mathbf{r},\mathbf{c},\mathbf{P}_{\mathrm{RS}}) \gets  \arg{\widehat{\mathcal{R}}_{\mathrm{RS}}^{\mathrm{MMSE}}(P)}$
%    %
%\Until$|r_{\mathrm{g}}^{(n)} - r_{\mathrm{g}}^{(n-1)} \big| < \epsilon$
%\end{algorithmic}
%\end{algorithm}
%%%%%%%%%%%%%%%%%%%%%%%%%%
%%

Each iteration of the algorithm increases the objective function, which is bounded above for a given power constraint, until convergence.
The global optimality of the limit point cannot be guaranteed  due to non-convexity.
However, the stationarity (KKT optimality) of the solution can be established using the arguments in \cite{Razaviyayn2014}.
%%
%%%
\section{Numerical Results and Analysis}
\label{Section_Numerical_Results}
%%%
In this section, we compare the performances of the different beamforming strategies.
We consider i.i.d channels with entries drawn from $\mathcal{CN}(0,1)$ and all performances are obtained by averaging over $100$ realizations.
We start by comparing the MMF rates achieved from 1) designated beamforming obtained by solving  \eqref{Eq_Problem_R_NoRS}, 2) single-stream degraded beamforming obtained by solving \eqref{Eq_Problem_R_Degraded_SS}, and 3) RS beamforming obtained by solving  \eqref{Eq_Problem_R_RS}.
Note that for problems \eqref{Eq_Problem_R_NoRS} and \eqref{Eq_Problem_R_Degraded_SS}, we plot the SDR upper-bounds (no randomization), hence presenting optimistic performances for the designated and degraded beamforming strategies.
On the other hand, the RS results are obtained by solving \eqref{Eq_Problem_WMSE_RS}, hence representing actual
achievable rates. All convex optimization problems are solved using the CVX toolbox \cite{Grant2008}.
In Fig. \ref{Fig_simulations_1}, we consider the same system  presented in Fig. \ref{Fig_NoRS_example}
with $K=6$ users divided over $M=3$ groups such that $G_{1} = 1$, $G_{2} = 2$ and $G_{3} = 3$.
The number of antennas is varied as $N = 2,4$ and $6$.
For $N=6$, it follows from Proposition \ref{Proposition_DoF_NoRS} and Proposition \ref{Proposition_DoF_RS} that full DoF is achieved by both the designated and RS beamforming schemes.
This is clear in Fig. \ref{Fig_simulations_1} where the two performances are almost identical, while degraded beamforming exhibits a DoF loss.
For $N = 4$, the system goes into the partially-overloaded regime and both schemes (designated and RS) achieve a MMF-DoF of $0.5$.
However, RS achieves a marginally improved MMF-rate performance.
This is due to the fact that the RS strategy offers more flexibility through the degraded stream, particularly when dealing with the interference caused by the largest group as explained in Section \ref{subsubsection_RS_insight}.
When $N$ drops to $2$, the system becomes fully-overloaded causing the MMF-rate of designated beamforming to saturate.
On the other hand, both the degraded and RS schemes achieve the same MMF-DoF of $1/3$,
with RS exhibiting gains in terms of the MMF-rate.
This is due to the designated streams in RS which provide a flavour of the scheme in \eqref{Eq_Problem_R_Degraded}, hence achieving an asymptotically constant gap with the scheme in \eqref{Eq_Problem_R_Degraded_SS}.
In Fig. \ref{Fig_simulations_2} we focus on fully-overloaded scenarios.
We consider a fixed number of antennas $N=4$ and a varied number of groups, i.e. $M = 3$ and $4$ groups with $2$ users per group.
It follows from the three propositions that the MMF-DoFs of the designated, degraded and RS schemes are $0$, $1/3$ and $1/2$ respectively for $M=3$,
and $0$, $1/4$ and $1/3$ respectively for $M=4$.
Such performances are exhibited in Fig. \ref{Fig_simulations_2}, where the pronounced gains achieved by the RS strategy in such overloaded scenarios are evident.
%%%%%%%%%%%%%%%%%%%%%%
\begin{figure}%[t!]
\centering
\includegraphics[width = 0.45\textwidth]{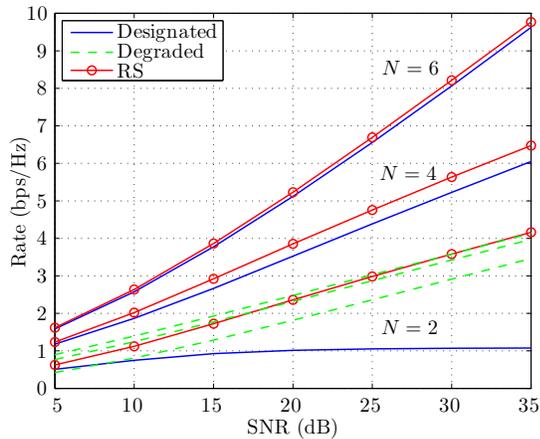}
\caption{MMF rate performances. $K=6$ users, $M=3$ groups, $G_{1},G_{2},G_{3}=1,2,3$ users, and $N = 2,4,6$ antennas.}
\label{Fig_simulations_1}
\end{figure}
%%%
%%%
\begin{figure}%[pbth]
\centering
\includegraphics[width = 0.45\textwidth]{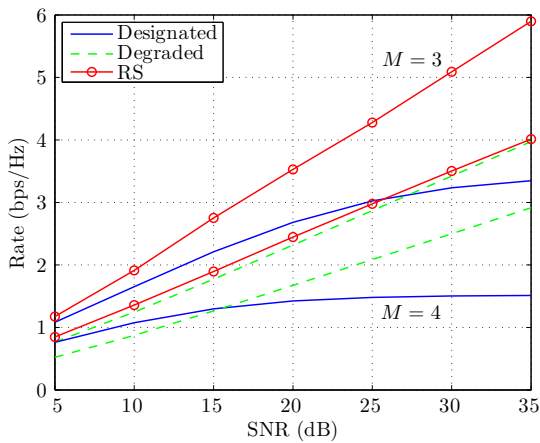}
\caption{MMF rate performances. $N = 4$ antennas and $M=3$ and $4$ groups of $2$ users each.}
\label{Fig_simulations_2}
\end{figure}
%%%
%%
Next, we look at the contributions of the designated and degraded parts in RS.
We consider the scenario in Fig. \ref{Fig_simulations_1} with $N=4$, and show the different contributions in Fig. \ref{Fig_simulations_3}.
As all groups achieve symmetric rates, the two contributions are inversely proportional for any given group.
According to the achievability scheme in Section \ref{subsubsection_achievability_RS}, it is sufficient from a DoF perspective to perform partitioned beamforming where groups $1$ and $2$ are served through designated beamforming, and group $3$ is served through degraded beamforming.
However, suppressing the designated part of group 3 is not necessarily optimum from a rate perspective.
In particular, Fig. \ref{Fig_simulations_3} shows that while the degraded parts of groups $1$  and $2$ have relatively small contributions, the designated part of group $3$ contributes significantly.
This follows by observing that designated beam of group $3$ can be useful by placing it in the null space of the two other groups.
%%
%%%
\begin{figure}%[pbth]
\centering
\includegraphics[width = 0.45\textwidth]{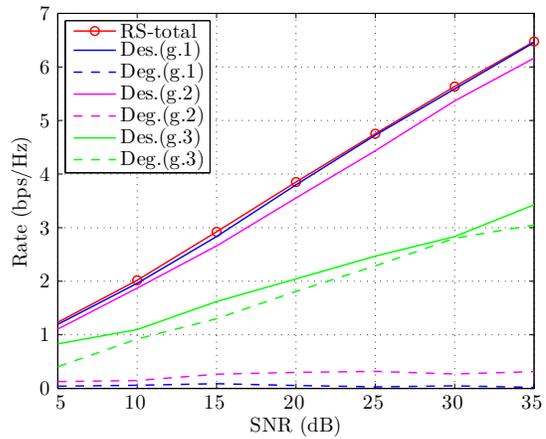}
\caption{RS rate contributions. $K=6$ users, $M=3$ groups, $G_{1},G_{2},G_{3} = 1,2,3$ users, and $N = 4$ antennas.}
\label{Fig_simulations_3}
\end{figure}
%%%
%%%
\section{Conclusion and Future Work}
\label{Section_Conclusion}
%%%
This paper considered the problem of MMF transmit beamforming in overloaded multigroup multicasting scenarios.
The limitations of the classical beamforming strategy in such scenarios were identified through DoF analysis.
Two alternative strategies have been proposed, namely the degraded strategy and the RS strategy.
From a DoF perspective, the RS strategy was shown to combine the benefits of the other two strategies, and surpass them both in some scenarios.
Simulation results showed that RS exhibits strictly higher MMF-rates in partially and fully overloaded scenarios.
This work focused more on proposing the strategies and analysing their DoF performances, and less on the design and optimization.
Indeed, the RS scheme can be further improved by incorporating the SIC structure in \eqref{Eq_Problem_R_Degraded} for the degraded part.
This poses extra optimization challenges due to the increased number of beamformers and the group ordering problem.
However, certain structures may be useful in ordering the groups such as the different group sizes for example.
Another interesting extension is the incorporation of imperfect Channel State Information at the Transmitter (CSIT) in the design and analysis.
This brings non-trivial challenges in terms of characterizing the DoF performance and also the robust optimization problem, as seen in \cite{Joudeh2016b,Joudeh2016a} for unicast beamforming.
However under imperfect CSIT, RS is expected to bring gains even when the system is not overloaded.
%%
%%%%%%%%%%%%%%%%%%%%%%%%%%%%%%%%%%%%%%%%%%%%%%%%%%%%
%%
%\appendices
\appendix
%%
%%%%%%%%%%%%%%%%%%%%%%%%%%%%%%%%%%%%%%%%%%%%%%%%%%%%%%%%%%%%%%%%
\subsection{Converse of Proposition \ref{Proposition_DoF_NoRS}}
\label{Appendix_Converse_NoRS}
%%%%%%%%%%%%%%%%%%%%%%%%%%%%%%%%%%%%%%%%%%%%%%%%%%%%%%%%%%%%%%%
Recall that a beamforming scheme is given by $\{\mathbf{P}(P)\}_{P}$ with one beamforming matrix for each power level.
For the $P$th level, the beamforming matrix is expressed by
\begin{equation}
\label{Eq_P_NoRS_converse}
\mathbf{P}(P) = \big[q_{1}(P)\mathbf{w}_{1}(P) \; \cdots \; q_{M}(P)\mathbf{w}_{M}(P) \big].
\end{equation}
In contrast to \eqref{Eq_P_NoRS_achievability}, the definition in \eqref{Eq_P_NoRS_converse} is more general as it allows the beamforming directions to change with $P$.
We define the $m$th power exponent as $a_{m} (P) = \frac{\log \big( q_{m} (P) \big)}{\log(P)}$,
where it is clear that $a_{m} (P) \in ( - \infty,1]$.
For the sake of analysis, the exponent is restricted to $[0,1]$ with no influence on the DoF result.
Moreover, we assume that the following limits are well-defined
\begin{equation}
\nonumber
\lim_{P \rightarrow \infty} a_{m} (P) = a_{m} \ \text{and} \ \lim_{P \rightarrow \infty} \mathbf{w}_{m}(P) = \mathbf{w}_{m}, \;
\forall m \in \mathcal{M}
\end{equation}
for all beamforming schemes\footnote{If such limits do not exist, then the limits in \eqref{Eq_user_DoF} and \eqref{Eq_group_DoF} may not exist. This is avoided by defining the DoF in \eqref{Eq_user_DoF} using the $\limsup$ which is guaranteed to exist due to \eqref{Eq_MMF_DoF_upperbound_multicast}. In turn, $a_{m}$ and $\mathbf{w}_{m}$ are taken as the values that achieve this limit superior, which are also guaranteed to exists as the sets of $a_{m} (P)$ and $\mathbf{w}_{m}(P)$ are compact, and from the extreme value theorem.}.
Note that $a_{m} = 0$ corresponds to $q_{m}(P) = O(1)$.
The maximum scaling factor amongst groups is denoted by $\bar{a} \triangleq \max_{m \in \mathcal{M}}a_{m}$.
Next we derive an upperbound for the achievable group-DoF.
We observe that for any $k \in \mathcal{K}$, $j \in \mathcal{M}$ and $P$, we have
$0 \leq |\mathbf{h}_{k}^{\Hrm}\mathbf{w}_{j}(P)|^{2} \leq \|\mathbf{h}_{k}^{\Hrm} \|^{2}$.
Such inner products characterize the inter-group interference nulling capabilities of different schemes.
For a given scheme $\{\mathbf{P}(P)\}_{P}$, we say that the $j$th beamformer interferes with the $m$th group (asymptotically) if
$|\mathbf{h}_{i}^{\Hrm}\mathbf{w}_{j}|^{2}  > 0$ for some $i \in \mathcal{G}_{m}$.
Let $\mathcal{I}_{m}$ be the subset of groups with beamformers interfering with the $m$th group, and let $\bar{a}_{m} \triangleq \max_{j \in \mathcal{I}_{m}}a_{j}$ be  the exponent of the dominant interferer.
Note that $\bar{a}_{m}=0$ for $\mathcal{I}_{m} = \emptyset$.
Moreover, there exist at least one
$i \in \mathcal{G}_{m}$ with SINR scaling as $ \gamma_{i}(P) = O\big( P^{a_{m} - \bar{a}_{m} }\big)$.
Recalling the DoF definition in \eqref{Eq_group_DoF}, we can write
\begin{equation}
\label{Eq_d_m_upper_bound_multicast}
d_{m} \leq \big( a_{m} - \bar{a}_{m}  \big)^{+}
\end{equation}
where $(\cdot)^{+}$ follows from the fact that the DoF is non-negative. We recall that for a given beamforming scheme, the achievable MMF-DoF satisfies $d \leq d_{m}$ for all $m \in \mathcal{M}$.
Now we argue that $d \leq d^{\star}$  for any feasible beamforming scheme.
$d \leq 1$ follows from \eqref{Eq_MMF_DoF_upperbound_multicast}. Hence, we focus on the two other cases.
For $d^{\star} = 0.5$, it is sufficient to show that $d \leq 0.5$ for $N = N_{M}-1$, as further decreasing $N$ cannot increase the DoF.
Since $N < N_{M}$, at least one group sees interference from $\mathbf{p}_{1}$ for any scheme.
Let $\mathcal{G}_{m_{1}}$ be a group that sees such interference, i.e. $1 \in \mathcal{I}_{m_{1}}$.
We may assume that $a_{m_{1}} >  \bar{a}_{m_{1}}$, as the contrary will limit the MMF-DoF to $0$ as seen from \eqref{Eq_d_m_upper_bound_multicast}.
Next, we write the following set of inequalities
\begin{align}
\label{Eq_d_multicast_UB_1}
d &\leq \frac{d_{1} + d_{m_{1}} }{2}\\
\label{Eq_d_multicast_UB_2}
& \leq \frac{ a_{1} + a_{m_{1}} - \bar{a}_{m_{1}} }{2} \\
\label{Eq_d_multicast_UB_3}
& \leq \frac{ a_{1} + a_{m_{1}} - a_{1} }{2} \leq 0.5.
\end{align}
\eqref{Eq_d_multicast_UB_1} follows by taking the average of two group-DoF as an upperbound for the minimum.
\eqref{Eq_d_multicast_UB_2} follows from \eqref{Eq_d_m_upper_bound_multicast}, while \eqref{Eq_d_multicast_UB_3} follows from $1 \in \mathcal{I}_{m_{1}}$ and  $a_{m_{1}}  \leq 1$.
Next we show that $d \leq 0$ for $N = N_{M-1}+G_{1}-1$.
Note that $N < 1 + K - G_{m}$ for all $m \in \mathcal{M}$.
Hence, each beamformer causes interference to at least one group.
Equivalently, we have $\bigcup_{m \in \mathcal{M}}\mathcal{I}_{m} = \mathcal{M}$.
For any power allocation with exponents $a_{1},\ldots,a_{M}$,
there exists at least one group that sees interference from $\mathbf{p}_{m_{1}}$, where $a_{m_{1}} = \bar{a}$.
Let the index of such group be $m_{2}$.
We have $d \leq d_{m_{2}} \leq  \big( a_{m_{2}} - \bar{a} \big)^{+}  = 0$, which completes the proof.
%%%%%%%%%%%%%%%%%%%%%%%%%%%%%%%%%%%%%%%%%%%%%%%%%%%%%%%%%%%%%%%%
%%
\subsection{Converse of Proposition \ref{Proposition_DoF_RS}}
\label{Appendix_Converse_RS}
%%%%%%%%%%%%%%%%%%%%%%%%%%%%%%%%%%%%%%%%%%%%%%%%%%%%%%%%%%%%%%%
For the RS strategy, we recall that a beamforming scheme is given by
$\big\{ \mathbf{P}_{\mathrm{RS}}(P) \big\}_{P}$, where
\begin{equation}
\label{Eq_P_RS_converse}
\mathbf{P}_{\mathrm{RS}}(P) = \big[q_{\mathrm{c}}(P)\mathbf{w}_{\mathrm{c}}(P) \ \  \mathbf{P}(P)\big]
\end{equation}
where $\mathbf{P}(P)$ is as described in the previous subsection.
On the other hand, we have $q_{\mathrm{c}}(P) = O\big(P_{\mathrm{t}}^{a_{\mathrm{c}}} \big)$,
with $a_{\mathrm{c}} \in [0,1]$ as the corresponding power scaling factor.
The DoF achieved by the degraded stream is defined as
$d_{\mathrm{c}} \triangleq \lim_{P \rightarrow \infty} \frac{R_{\mathrm{c}}(P) }{\log_{2}(P)}$.
As the degraded stream is decoded by all receivers while treating designated streams as noise,
the degraded  DoF is upperbounded by
\begin{equation}
\label{Eq_d_c_upper_bound_multicast}
d_{\mathrm{c}} \leq \big( a_{\mathrm{c}} - \bar{a} \big)^{+} \leq 1 - \bar{a}
\end{equation}
which is limited by the maximum power scaling across all designated beams.
The fraction of $d_{\mathrm{c}}$ allocated to the $m$th group is denoted by $c_{m}$, where $\sum_{m=1}^{M}c_{m} = d_{\mathrm{c}}$.
Hence, the $m$th group-DoF is given by $c_{m} + d_{m}$, consisting of a common part and a designated part.
The  MMF-DoF achieved by a given scheme satisfies $d_{\mathrm{RS}} \leq c_{m}+d_{m}$ for all $m \in \mathcal{M}$.
We recall that  the maximum number of groups that can be served with interference-free beamforming is denoted by $M_{\mathrm{D}}^{\star}$, which is expressed in \eqref{Eq_M_D_star}.
Hence, for any feasible precoding scheme, at least $M_{\mathrm{c}}^{\star} = M - M_{\mathrm{D}}^{\star}$ groups receive non-zero interference from the designated beams.
In this proof, we show that $d_{\mathrm{RS}} \leq \frac{1}{1+M_{\mathrm{c}}^{\star}}$.
Before we proceed, we present the following result which plays an important role in the proof.
\begin{lemma}
\label{Lemma_bound_groups_multicast}
\textnormal{
For any designated beamforming matrix $\mathbf{P}$,  $\mathbf{p}_{1}$ interferes with at least
$M_{\mathrm{c}}^{\star}$ groups.
Moreover, each of $\mathbf{p}_{2},\ldots,\mathbf{p}_{M}$ interfere with at least $M_{\mathrm{c}}^{\star}-1$ groups.}
\end{lemma}
\begin{proof}
From the discussion in Section \ref{Subsection_MMF_DoF_NoRS},
it follows that placing $\mathbf{p}_{m}$ in the null space of all groups in the subset $\mathcal{L}_{m} \subseteq \mathcal{M} \setminus m$ requires that
\begin{equation}
\label{Eq_N_m_condition}
N \geq 1 + \sum_{j \in \mathcal{L}_{m} } G_{j}.
\end{equation}
Hence, finding the minimum number of groups $\mathbf{p}_{m}$ interferes with is equivalent to finding $\mathcal{L}_{m}$ with the maximum $|\mathcal{L}_{m}|$ such that \eqref{Eq_N_m_condition} is satisfied.
First, we observe that for fixed $|\mathcal{L}_{m}| = L$, the subset of groups requiring the least number of antennas to satisfy \eqref{Eq_N_m_condition} is given by
\begin{align}
\label{Eq_L_m}
\mathcal{L}_{m}^{\star}(L) =
\begin{cases}
 \{1,\ldots,m-1,m+1,\ldots,L+1\},  m      \leq    L \\
 \{1,\ldots,L\},   m     >   L
\end{cases}
\end{align}
which follows from \eqref{Eq_Group_order} and \eqref{Eq_N_m_condition}.
Hence, having $N < 1 + \sum_{j \in \mathcal{L}_{m}^{\star}(L) } G_{j}$
implies that $\mathbf{p}_{m}$ cannot be placed in the null space of any subset of $L$ groups,
Equivalently, $\mathbf{p}_{m}$ interferes with at least $M - L - 1$ groups (by excluding $m$).
To characterize this condition for all $m \in \mathcal{M}$, we write
\begin{equation}
\label{Eq_Nt_Mms}
N = N_{M_{\mathrm{D}}^{\star}} + \bar{N} = 1+ \sum_{j = 2}^{M_{\mathrm{D}}^{\star}} G_{j} + \bar{N}
\end{equation}
%%%
where $0 \leq \bar{N} < G_{M_{\mathrm{D}}^{\star} + 1}$.
This follows directly from \eqref{Eq_M_D_star}.
Now looking at $m=1$, the corresponding beamformer can be placed in the null space of at most $M_{\mathrm{D}}^{\star} - 1$ groups, i.e. groups
$\{2,\ldots,M_{\mathrm{D}}^{\star}\}$.
This follows by observing that
\begin{equation}
\label{Eq_N_m_1}
1 + \sum_{j = 2}^{M_{\mathrm{D}}^{\star}}  G_{j} \leq N
<  1 + \sum_{j = 2}^{M_{\mathrm{D}}^{\star}+1}  G_{j}.
\end{equation}
From the right-most term in \eqref{Eq_N_m_1}, we can see that $\mathbf{p}_{1}$ causes interference to at least $M_{\mathrm{c}}^{\star}$ groups.
Next, we consider the groups $m \in \{2,\ldots,M_{\mathrm{D}}^{\star} \}$.
We can write
\begin{equation}
\label{Eq_Nt_inequality_appendix}
N  = 1 + \sum_{j = 1,j\neq m}^{M_{\mathrm{D}}^{\star}}  G_{j} + (G_{m}- G_{1}) + \bar{N}
< 1 + \sum_{j = 1,j\neq m}^{M_{\mathrm{D}}^{\star} + 2}  G_{j}
\end{equation}
where \eqref{Eq_Nt_inequality_appendix} follows from $\bar{N} < G_{M_{\mathrm{D}}^{\star}+1}$ and
$G_{m}- G_{1} < G_{M_{\mathrm{D}}^{\star}+2}$.
Hence, in the best case scenario, $\mathbf{p}_{m}$ is placed in the null space of groups
$\{1,\ldots,m-1,m+1,\ldots,M_{\mathrm{D}}^{\star}+1\}$,
and causes interference to the remaining $M_{\mathrm{c}}^{\star} - 1 $.
Finally, we consider $m \in \{M_{\mathrm{D}}^{\star}+1,\ldots,M \}$.
Here, the best scenario occurs when $\bar{N} \geq G_{1}$, from which we can write
\begin{equation}
1 + \sum_{j = 1}^{M_{\mathrm{D}}^{\star}}  G_{j} \leq N <  1 + \sum_{j = 2}^{M_{\mathrm{D}}^{\star}+1}  G_{j}.
\end{equation}
It can be seen that in this case also, $\mathbf{p}_{m}$ causes interference to at least $M_{\mathrm{c}}^{\star} - 1 $ groups.
\end{proof}
Since the minimum group-DoF is upper-bounded by the average of any subset of group-DoFs, by taking the subset $\mathcal{S} \subseteq \mathcal{M}$, we can write
\begin{equation}
d_{\mathrm{RS}} \leq \frac{\sum_{m \in \mathcal{S}} c_{m} + d_{m}}{|\mathcal{S}|}  \leq
\frac{d_{\mathrm{c}} + \sum_{m \in \mathcal{S}} d_{m}}{|\mathcal{S}|}.
\end{equation}
where the right-hand side inequality follows from the fact that
$\sum_{m \in \mathcal{S}} c_{m} \leq \sum_{m \in \mathcal{M}}c_{m} = d_{\mathrm{c}}$.
Now, we need to find the \emph{right} subset $\mathcal{S}$ which gives us a meaningful upper-bound in closed-form, that applies to any feasible precoding scheme.
Let $\bar{m} \in \mathcal{M}$ be the index of the group with the largest power scaling, i.e. $a_{\bar{m}} = \bar{a}$.
Moreover, let $\mathcal{S}_{\bar{m}} \subseteq \mathcal{M} \setminus \bar{m}$ be the set of groups that see interference from $\mathbf{p}_{\bar{m}}$.
From Lemma \ref{Lemma_bound_groups_multicast}, we know that $|\mathcal{S}_{\bar{m}}| \geq M_{\mathrm{c}}^{\star} - 1$.
For the upper-bound, we assume that $|\mathcal{S}_{\bar{m}}| = M_{\mathrm{c}}^{\star} - 1$, as increasing the number of groups that see interference does not increase the DoF.
We also assume that $\bar{m} \neq 1$, as the contrary does not influence the result as we see later.
Since $\mathbf{p}_{1}$ interferes with at least $M_{\mathrm{c}}^{\star}$ groups (from Lemma \ref{Lemma_bound_groups_multicast}), we have at least one group that sees interference from $\mathbf{p}_{1}$ and is not in $\mathcal{S}_{\bar{m}}$. Let the index of such group be $m_{1}$.
From \eqref{Eq_d_m_upper_bound_multicast}, note that $a_{1} \geq a_{m_{1}}$ implies $d_{1} + d_{m_{1}} \leq a_{1}$, while
$a_{1} \leq a_{m_{1}}$ implies $d_{1} + d_{m_{1}} \leq a_{m_{1}}$.
We assume, without loss of generality, that $a_{1} \geq a_{m_{1}}$, as the contrary does not affect the result.
The set of groups for the upper-bound is taken as $\mathcal{S} = \{1,m_{1}, \mathcal{S}_{\bar{m}}\}$ with $|\mathcal{S}| = M_{\mathrm{c}}^{\star}+1$.
Now, we can write
\begin{align}
\label{Eq_d_RS_ub_multicast_1}
d_{\mathrm{RS}} & \leq \frac{ d_{\mathrm{c}} +  d_{1} +d_{m_{1}} + \sum_{m \in \mathcal{S}_{\bar{m}}} d_{m}}{M_{\mathrm{c}}^{\star}+1} \\
\label{Eq_d_RS_ub_multicast_2}
& \leq \frac{ 1 -  \bar{a} + a_{1}  }{M_{\mathrm{c}}^{\star}+1} \leq \frac{1 }{M_{\mathrm{c}}^{\star}+1}
\end{align}
where the left-hand side inequality in \eqref{Eq_d_RS_ub_multicast_2} follows from the fact that $d_{m} = 0$ for all $m \in \mathcal{S}_{\bar{m}}$ and \eqref{Eq_d_c_upper_bound_multicast}, and the right-hand side inequality follows from $\bar{a} \geq a_{1}$.
Note that if we assume that $\bar{m} = 1$, then we can also assume that $|\mathcal{S}_{\bar{m}}| = M_{\mathrm{c}}^{\star}$. As a result, the same upper-bound holds by adding $m_{1}$ to $\mathcal{S}_{\bar{m}}$ and setting $d_{m_{1}} = 0$.
This completes the proof.
%%%%%%%%%%%%%%%%%%%%%%%%%%%%%%%%%%%%%%%%%%%%%%%%%%%%
\ifCLASSOPTIONcaptionsoff
  \newpage
\fi
\footnotesize
\bibliographystyle{IEEEtran}
\bibliography{References}

%\balance

\end{document}